\documentclass[11pt]{article}

\usepackage[margin=1in]{geometry}
\usepackage{makecell}
\usepackage[utf8]{inputenc}
\usepackage{url,hyperref}
\usepackage{fancyvrb}
\usepackage{booktabs}
\usepackage{multirow}
\usepackage[table, dvipsnames]{xcolor}
\usepackage{subcaption}
\usepackage{subcaption}
\usepackage{float}
\usepackage{subcaption}
\usepackage{algpseudocode}

\usepackage[ruled]{algorithm2e}
\usepackage{lipsum} 
\usepackage{tablefootnote} 
\usepackage{stfloats}
\usepackage{makecell}
\usepackage{amsmath,amssymb,amsthm}
\usepackage{graphicx}

\newcommand{\eps}{\epsilon}
\newcommand{\veps}{\varepsilon}
\newcommand{\etal}{{\it{et al.}}}

\newcommand{\E}{\mathcal{E}}

\newcommand{\EE}{\mathbb{E}}

\newcommand{\X}{\mathcal{X}}

\newcommand{\distkm}{d^{est}_{km}}

\newtheorem{thm}{Theorem}[section]
\newtheorem{lem}{Lemma}[section]
\newtheorem{remark}{Remark}[section]
\newtheorem{claim}{Claim}[section]
\newtheorem{fact}{Fact}[section]
\newtheorem{defn}{Definition}[section]
\newtheorem{definition}{Definition}[section]

\newcommand{\tp}{\Tilde{\Phi}}

\newcommand{\Tds}{\widetilde{D^2}}

\newcommand{\Wot}{\widetilde{WO}}
\newcommand{\So}{{SO}}

\usepackage{authblk}
\usepackage{booktabs}
\usepackage{array}
\usepackage{hyperref}
\usepackage[round]{natbib}
\usepackage{times}  

\hypersetup{
    colorlinks=true,
    linkcolor=blue,
    citecolor=blue,
    urlcolor=blue
}

\title{Improved Algorithms for Clustering with Noisy Distance Oracles}

\author[1]{Pinki Pradhan\thanks{pinki.pradhan@niser.ac.in}}
\author[1]{Anup Bhattacharya\thanks{anup@niser.ac.in}}
\author[2]{Ragesh Jaiswal\thanks{rjaiswal@cse.iitd.ac.in}}

\affil[1]{NISER, Bhubaneswar, an OCC of HBNI, India}
\affil[2]{IIT Delhi, India}

\date{}  

\begin{document}

\maketitle

\begin{abstract}
Bateni~\etal~has recently introduced the {\em weak-strong distance oracle model} to study clustering problems in settings with limited distance information. Given query access to the strong-oracle and weak-oracle in the weak-strong oracle model, the authors design approximation algorithms for $k$-means and $k$-center clustering problems. In this work, we design algorithms with improved guarantees for $k$-means and $k$-center clustering problems in the weak-strong oracle model. The $k$-means++ algorithm is routinely used to solve $k$-means in settings where complete distance information is available. One of the main contributions of this work is to show that $k$-means++ algorithm can be adapted to work in the weak-strong oracle model using only a small number of strong-oracle queries, which is the critical resource in this model. In particular, our $k$-means++ based algorithm gives a constant approximation for $k$-means and uses $O(k^2 \log^2{n})$ strong-oracle queries. This improves on the algorithm of Bateni~\etal\ that uses $O(k^2 \log^4n \log^2 \log n)$ strong-oracle queries for a constant factor approximation of $k$-means. For the $k$-center problem, we give a simple {\em ball-carving} based $6(1 + \eps)$-approximation algorithm that uses $O(k^3 \log^2{n} \log{\frac{\log{n}}{\eps}})$ strong-oracle queries. This is an improvement over the $14(1 + \eps)$-approximation algorithm of Bateni~\etal\ that uses $O(k^2 \log^4{n} \log^2{\frac{\log{n}}{\eps}})$ strong-oracle queries.  To show the effectiveness of our algorithms, we perform empirical evaluations on real-world datasets and show that our algorithms significantly outperform the algorithms of Bateni~\etal

\end{abstract}

\section{INTRODUCTION}

Clustering problems such as the $k$-means, $k$-median and $k$-center are often studied in {\em complete distance information} settings where the embeddings of the $n$ points to be clustered are given as input to the problem. These $k$-clustering problems are commonly formulated to minimize some cost function of the input. Often these formulations turn out to be $\mathsf{NP}$-hard, and approximation algorithms are designed to solve these problems. However, for many machine learning applications, the assumption of complete knowledge of the dataset under consideration might be either infeasible or very expensive to meet. This motivates the study of clustering problems in {\em partial distance information} settings, where there is a tradeoff between the accuracy of the input data and the resulting cost of clustering. The more accurate you want the clustering solution to be, the cost of collecting accurate information about the dataset would become higher. Questions such as {\em what is the best computational guarantee (e.g., approximation factor) that one can obtain within a budget}, or {\em what is the minimum cost for obtaining a target guarantee}, become meaningful. Such studies can also be interpreted as computing with noisy information and are not entirely new (e.g., \cite{FRPU1994}), and the problems considered in these settings range from sorting and searching to graph problems \citep{GX2023,KK2007}.

Clustering problems have been studied in various partial information settings. Often these works are formulated as having access to an oracle that provides some meaningful information about the underlying clustering, and the goal is to obtain a good approximate clustering solution using only a small number of oracle queries. Several works studied variants of this problem with different assumptions about the oracle and the clustering objective. \cite{LM2022} give tight query complexity bounds for the cluster recovery problem with a membership-query oracle, which on a query with any two points, answers whether the points belong to the same cluster or not. Several works \citep{ABJK2018,MS2017} have studied clustering using a noisy membership-query oracle, also called a same-cluster query oracle, to get approximate solutions. \cite{GLNS2022} study algorithms for $k$-means where noisy labels are provided for each input point using an adversarial or random perturbation of the labels.

In this work, we study clustering problems in the {\em weak-strong oracle model} introduced by \cite{BDJW2024}. The weak-strong oracle model has access to an expensive strong-oracle giving exact distances between any two points and a cheaper weak-oracle that gives noisy answers for the distances. Given that meaningful vector embeddings of varied accuracy are now frequently available in most machine learning applications for various types of data, it makes sense to consider a noise model where the distances (indicating dissimilarity) between data items are available, and accurate information comes at a cost. This was the primary motivation behind the weak-strong oracle model given by \cite{BDJW2024}. 
Next, we describe clustering problems with access to distance oracles.


In the clustering with distance oracles problem, the algorithm is given as input the number $n$ of points in the dataset and the number $k$ of clusters. We assume that the input points $X$ belong to some metric space $(\X,d)$. However, we don't have direct access to their embeddings, and also we don't know the distances between any two of these points. We are given query access to distance oracles that, on a query with any two points in $\X$, respond with the distance between the two points. 
Each query has a cost associated with it, and we assume the strong-oracle queries to be much more expensive compared to the weak-oracle queries. We design algorithms for clustering problems that make only a small number of queries to these distance oracles and return good approximate solutions. The first distance oracle we mention is the strong-oracle $\So$ which returns the exact distance $d(x,y)$ for any two query points $x,y\in \X$.
\begin{definition}[Strong-oracle model] In this model, an algorithm is given access to a strong-oracle that for any $x,y\in \X$, returns $\So(x,y)=d(x,y)$, the exact distance between the two points. 
\end{definition}
We first observe that $O(nk)$ strong-oracle queries suffice to run the $k$-means++ algorithm, and using known results, one gets a constant factor bi-criteria approximation for $k$-means \citep{ADK2009} using $O(nk)$ strong-oracle queries. We also observe that $O(nk)$ strong-oracle queries are enough for a $2$-approximation for $k$-center problem \citep{G1986}. However, as mentioned earlier, access to such strong-oracles might be limited and quite expensive in practice. This motivates us to design algorithms that access a much cheaper and possibly erroneous weak-oracle $\Wot$. For some fixed $\delta\in (0,1/2)$, given two points $x,y\in \X$, the weak-oracle $\Wot$ returns the exact distance $d(x,y)$ between the two points independently with probability $(1-\delta)$ and returns an arbitrary answer with the remaining probability. Moreover, we also assume that the query answers of the weak-oracle $\Wot$ are \textit{persistent} in the sense that the weak-oracle always returns the same answer to a query, even if asked multiple times.

\begin{definition} [Weak-oracle model]
In this model, an algorithm has access to a weak-oracle $\Wot$ that for any two points $x,y\in \X$, outputs $d(x,y)$ with probability $1-\delta$, and with remaining probability $\delta$, outputs an arbitrary value. \end{definition}

\textbf{Weak-strong oracle model} of \cite{BDJW2024}: An algorithm in this model is allowed to query both the weak-oracle and the strong-oracle. We show that algorithms with good guarantees can be designed for clustering problems in this model that use only a small number of strong-oracle queries. 

We note here that the assumption on the weak-oracle that it returns the exact distance between any two points with constant probability might appear to be quite {\em strong} for applications in practice. In Appendix \ref{sec:just-weak}, we give some informal justification for the weak-strong oracle model arguing that some natural attempts to model noisy oracle distances can be thought of as instantiations of the weak-oracle model.


\subsection{Clustering Problems}
Let $X$ be a set of $n$ points in a metric space $(\X,d)$. The $k$-means cost of $X$ with respect to any set $C\subset \X$ is given as $\Phi(X,C)=\sum_{x\in X} d^2(x,C)$, where $d(x,C)=\min_{c\in C} d(x,c)$. The $k$-means clustering problem is defined as follows:
{\em Given a set $X$ of $n$ points in a metric space $(\X,d)$ and an integer $k$, the objective in the $k$-means clustering problem is to output a set $C$ of $k$ centers for which the $k$-means cost function $\Phi(X,C)$ is minimized.}
The $k$-means clustering problem is very well studied theoretically and has numerous applications in practice. 
This problem is known to be $\mathsf{NP}$-hard \citep{D2008}. 
The $k$-means++ algorithm is often used to solve $k$-means in practice \citep{AV2007}. 
The main idea is to iteratively sample $k$ points using a distribution called $D^2$-distribution that is updated after every iteration and depends on the previously chosen centers.
It is known that $k$-means++ gives $O(\log k)$-approximation in expectation for $k$-means \citep{AV2007}. Moreover, it can also be adapted to obtain a constant factor bi-criteria approximation \citep{ADK2009} simply by oversampling centers (i.e., sample more than $k$ centers). 
We shall refer to this as oversampling-$k$-means++.
An $(\alpha,\beta)$-bi-criteria approximation uses at most $\alpha k$ centers and returns a solution of cost at most $\beta$ times the optimal $k$-means solution.
The $(k,z)$-clustering problem is a generalization of the $k$-means problem. 
In the $(k,z)$-clustering problem, given a set $X$ of $n$ points from a metric space and an integer $k$, the objective is to return a set $C\subset X$ of $k$ centers for which $\sum_{x\in X} d(x,C)^z$ is minimized. For $z=2$, this is the $k$-means problem and for $z=1$, this is the $k$-median problem.

Another way to formulate data clustering is using the $k$-center problem, where the goal is to optimize the following cost function. The $k$-center cost of $X$ with respect to a set $C\subset \X$ is defined as $\phi(X, C) = \max_{x \in X} d(x, C)$, where $d(x, C) = \min_{c \in C} d(x, c)$. 
More specifically, the $k$-center problem is defined as follows:
{\em Given a set $X$ of $n$ points in a metric space $(\X,d)$ and an integer $k$, the objective in the $k$-center problem is to output a set $C$ of $k$ centers for which the $k$-center cost function $\phi(X,C)$ is minimized.}
In other words, the goal is to find $k$ points such that the radius of the balls centered at these points that cover $X$, is minimized.
The $k$-center problem is known to be $\mathsf{NP}$-hard and $2$-approximation algorithms are known for the $k$-center problem \citep{G1986}.


\subsection{Main Results} 

\paragraph{Weak-strong oracle model:} 
We design algorithms for $k$-means and $k$-center problems in the weak-strong oracle model that use only a small number of strong-oracle queries while giving good approximate solutions. Our algorithms provide better guarantees compared to the work of \cite{BDJW2024}. We next state our main results for $k$-means and $k$-center problems, and provide more details in Sections \ref{sec:k-means-weak-strong-main} and \ref{k-center-weak-strong-main}, respectively.

\begin{thm}(Upper bound for $k$-means) \label{thm:k-means-cons-approx} Let $\eps \in (0,1)$ and $\delta \leq 1/3$. There exists a randomized algorithm for the $k$-means problem that adapts oversampling $k$-means++ in the weak–strong oracle model and outputs a $\big(O\big(\frac{ \log n}{\eps^4}\big), 40(1+\eps)\big)$ bi-criteria approximation for $k$-means. The algorithm uses $O\big(\frac{k^2 \log^2 n}{\eps^8}\big)$ strong-oracle queries and $O\big(\frac{nk \log n}{\eps^4}\big)$ weak-oracle queries, and succeeds with constant probability. \end{thm}

\begin{thm}(Upper bound for $k$-center)\label{thm:k-center} Let $\eps>0$. There exists a $6(1+\eps)$-approximation algorithm for the $k$-center problem and that makes $O \left(k^3 \log^2{n} \log{\frac{\log{n}}{\eps}} \right)$ strong-oracle  and $O\left(nk\log{n} \log{\frac{\log{n}}{\eps}} \right)$ weak-oracle queries in the weak-strong oracle model and succeeds with probability at least $(1-1/n^4)^2$. \end{thm}

\begin{remark}\label{remark:1} We obtain these results for the failure probability $\delta=1/3$. We note that our algorithms can be extended to work with any $0<\delta<1/2$. The dependency of $\delta$ on the sample complexity would become ${1}/{(1/2-\delta)^2}$ \cite{BDJW2024}. \end{remark}

\paragraph{Weak-oracle model:} An algorithm in this query model can only access the weak-oracle $\Wot$. We show a lower bound result in the weak-oracle model proving that any constant factor approximate clustering solution for $k$-means in this model requires $\Omega({nk}/{(1-2\delta)^2})$ weak-oracle queries. 
We prove the following result in Appendix \ref{sec:weak-oracle}.

\begin{thm}\label{thm:lower-bound} Let $\delta \in (0,1/2)$. Any randomized algorithm giving a $O(1)$-approximation for $k$-means in weak-oracle model with probability at least $3/4$ requires $\Omega({nk}/{(1-2\delta)^2})$ queries in expectation. \end{thm}

\paragraph{Our contributions and comparisons with known works} In this work, we provide improved algorithms for clustering problems in the weak-strong oracle model of \cite{BDJW2024}. We highlight the main contributions of our work and provide a detailed comparison of our results with \cite{BDJW2024}.

\begin{itemize}
    \item $k$-means++ algorithm is widely used to solve $k$-means in practice. $k$-means++ algorithm has also been extended to work in the presence of outliers \citep{BVX2019,DP2023}, and even when there are errors in the computation of the $D^2$-sampling distribution \citep{BERS2020,GOR2023}. However, it was not known whether $k$-means++ works with unreliable distance estimates as in the weak-strong model. 
    We show that even though the distance estimates given by the weak-oracle are unreliable, it is possible to come up with an alternate distance measure with respect to which one can run the $k$-means++ algorithm to obtain a good approximation guarantee for $k$-means. We believe that one of the main contributions of this work is to adapt the $k$-means++ algorithm to work in this restricted model. 
    \item \cite{BDJW2024} give a constant factor approximation algorithm for $k$-means problem using $O(k^2 \log^4 n \log^2({\log n}/{\eps}))$ strong-oracle queries and $O(nk \log^2 n \log({\log n}/{\eps}))$ weak-oracle queries. Our algorithm for $k$-means uses $O(\frac{k^2\log^2 n}{\eps^8})$ strong-oracle queries and $O(\frac{nk\log n}{\eps^4})$ weak-oracle queries and gives a constant approximation for $k$-means, with a better constant factor.
    \item Our algorithm for $k$-means in the weak-strong oracle model can be generalized to obtain a $O(2^{2z})$ factor approximation in expectation for the $(k,z)$ clustering problem.
    \item \cite{BDJW2024} adapts Meyerson's sketch for online facility location for $k$-means problem in the weak-strong model \cite{M2001}. Their algorithm guesses the optimal value of $k$-means solution, and using a bounded aspect-ratio assumption on the input instance, they show that the number of such guesses is limited. Our algorithm doesn't use any such assumption.    
    \item For the $k$-center problem, we obtain the following improvements over \cite{BDJW2024}. We give a $6(1+\eps)$-approximation algorithm for $k$-center using $O(k^3 \log^2{n} \log{{\log{n}}/{\eps}})$ strong-oracle queries and $O(nk\log{n} \log{{\log{n}}/{\eps}})$ weak-oracle queries. This is in contrast to \cite{BDJW2024} who give a $14(1+\eps)$-approximation using $O(k^2 \log^4 n \log^2({\log n}/{\eps}))$ strong-oracle and $O(nk \log^2 n \log ({\log n}/{\eps}))$ weak-oracle queries.
\end{itemize}

\paragraph{Experiments} We experimentally verify the performance of our algorithms for $k$-means and $k$-center problems on synthetic as well as real-world datasets. For experiments on synthetic datasets, we use datasets generated using the stochastic block model \citep{MNS2015,A2018}. For experiments on real-world datasets, we use MNIST dataset \cite{D2012} with SVD and t-SNE embeddings \citep{MH2008}. We compare the experimental results of our algorithms with those of \cite{BDJW2024}, and note that our algorithms use significantly fewer number of queries to output clusterings of comparable cost. Our algorithms for $k$-means and $k$-center use at least $61\%$ and $23\%$ fewer queries, respectively, compared to \cite{BDJW2024}.
The experimental details are provided in Section \ref{sec:experiments-main} and Appendix \ref{sec:detailed-experiments}.

\begin{table*}[t]
\centering
\caption{Comparison of approximation guarantees and query complexities for $k$-means and $k$-center problems with \cite{BDJW2024} in the weak-strong oracle model. For $k$-means, we report  the bi-criteria approximation guarantees obtained by our algorithm as well as \cite{BDJW2024}.}
\label{table:comparison}
\begingroup 

\resizebox{\textwidth}{!}{  
\begin{tabular}{|c|c|c|c|c|}
\hline
\textbf{Problems} & \textbf{Results} & \textbf{Approximation Guarantee} & \textbf{Strong-oracle Queries} & \textbf{Weak-oracle Queries} \\ \hline

\multirow{2}{*}{$k$-means}  
& \cite{BDJW2024} 
& $(O(\log^2 n \log(\frac{\log n}{\epsilon})), O(1))$\footnotemark[1]
& $O(k^2 \log^4 n \log^2(\frac{\log n}{\epsilon}))$  
& $O(nk \log^2 n \log(\frac{\log n}{\epsilon}))$ 
\\ \cline{2-5}

& Our 
& $(O(\frac{\log n}{\epsilon^4}), 40(1+\epsilon))$ 
& $O\left(\frac{k^2 \log^2 n}{\epsilon^8}\right)$  
& $O\left(\frac{nk \log n}{\epsilon^4}\right)$ 
\\ \hline

\multirow{2}{*}{$k$-center} 
& \cite{BDJW2024} 
& $14(1+\epsilon)$ 
& $O(k^2 \log^4 n \log^2(\frac{\log n}{\epsilon}))$ 
& $O(nk \log^2 n \log (\frac{\log n}{\epsilon}))$  
\\ \cline{2-5}

& Our 
& $6(1+\epsilon)$ 
& $O\left(k^3 \log^2 n \log\left(\frac{\log n}{\epsilon}\right)\right)$ 
& $O\left(nk \log n \log\left(\frac{\log n}{\epsilon}\right)\right)$ 
\\ \hline

\end{tabular}
}
\endgroup

\end{table*}
\footnotetext[1]{\cite{BDJW2024} only mention a $O(1)$-approximation factor. We estimate the constant factor to be much larger than $40$.}

\subsection{Technical Overview} In this section, we describe the main ideas of our algorithm for $k$-means in the weak-strong oracle model. One of the most widely used algorithms for $k$-means is the $k$-means++ seeding algorithm that works in $k$ iterations. In the first iteration, it chooses the first center as a point of the input sampled uniformly at random. Each of the remaining $(k-1)$ iterations samples a point following a non-uniform and adaptive sampling distribution known as the $D^2$ distribution that depends on the previously chosen centers. Sampling a point from the $D^2$-distribution is called $D^2$-sampling.
The $k$ points sampled using $D^2$-sampling are returned as the solution for $k$-means. 
The $D^2$-distribution is such that the probability of choosing a point $x$ is proportional to the squared distance $d(x, C)^2$, where $C$ is a set of chosen centers and $d(x, C)=\min_{c\in C} d(x,c)$. \cite{ADK2009} showed a constant factor bi-criteria approximation for the oversampling version of $k$-means++ in which $O(k)$ iterations of $D^2$-sampling are used.

One of the main contributions of this work is to show that one can adapt $k$-means++ to the weak-strong oracle model. It is easy to observe that one can simulate $k$-means++ in the strong-oracle model using $O(nk)$ strong-oracle queries. However, things become complicated in the weak-strong model, in which the distances returned by the weak-oracle cannot be trusted. To construct the $D^2$-sampling distribution, one must calculate $d(x, C)$ for any point $x$ and any set $C$ of centers. Since $d(x, C)=\min_{c\in C} d(x,c)$, if we want to compute this distribution exactly, it would require $\Omega(nk)$ strong-oracle queries. Since we don't want to use too many strong-oracle queries, we proceed as follows. Using weak-oracle query answers, we compute an estimate $\distkm(x, C)$ such that with high probability, $d(x,C)$ and $\distkm(x,C)$ are within some small additive factor. We construct the sampling distribution to be used by the algorithm using these $\distkm(x, C)$ values for all points $x$.

To compute $\distkm(x, C)$ for all $x \in X$, we use the fact that weak-oracle distance queries are incorrect independently with probability $\delta < 1/2$. Consider any $c\in C$. Let us obtain an estimate on $d(x,c)$ using only weak-oracle queries. Let $B(c,r_c)$ denote a ball around the center $c$ of radius $r_c$, and suppose that $r_c$ is chosen such that $B(c,r_c)$ contains $\Omega(\log n)$ points. Then, the key idea is that the median of the weak-oracle distance queries on pairs $(x, y)$ with $y\in B(c,r_c)$ is a good estimate on $d(x, c)$. More formally, we use $\distkm(x,c)=median\{\Wot(x,y):y\in B(c,r_c)\}$ as an estimate for $d(x, c)$. We note that \cite{BDJW2024} also uses the median of query answers to estimate $d(x,c)$. However, their algorithm used these values differently, as we will see in the remaining discussion. 
Since each weak-oracle query answer is wrong independently with probability at most $1/2$, using Chernoff bounds, one can show that for sufficiently large-sized ball $B(c,r_c)$, with high probability, $|d(x,c)-\distkm(x,c)|\leq r_c$. Now, to approximate $d(x,C)$, we compute $\distkm(x,c)+r_c$ for all $c\in C$ and set $\distkm(x,C)=\min_{c\in C} \{\distkm(x,c) + r_c\}$. Once we have the sampling distribution $\Tds$ for which a point $x$ is sampled with probability proportional to $\distkm(x,C)^2$, we show that the approach of \cite{ADK2009} can be adapted to obtain a constant factor bi-criteria approximation for $k$-means that succeeds with constant probability. Even though the analysis becomes non-trivial and we incur an approximation loss, the high-level analysis ideas of \cite{ADK2009}  goes through.

The main idea is to run the oversampling $k$-means++ where we iteratively sample and update the center set $C$. However, the sampling is with respect to $\distkm(x, C)^2$ instead of $d(x, C)^2$. The goal is to pick a good center from every optimal cluster. Let us see why oversampling using $\distkm(x, C)^2$ will be sufficient to find good centers from every optimal cluster, and hence obtain a good approximation guarantee. For the simplicity of discussion, let us first assume that every optimal cluster has $\Omega(\log{n})$ points. We will later see how to drop this assumption. Let us draw a parallel with oversampling $k$-means++ that samples using accurate distance values. Why does it manage to pick good centers from every optimal cluster? Consider an intermediate center set $C$. Certain optimal clusters will have a good representative center in $C$, whereas others remain `uncovered'. Sampling using the $D^2$ distribution boosts the probability of sampling from an uncovered cluster, so there is a good chance that the next sample belongs to an uncovered cluster.
Moreover, we can argue that given the next sample is from an uncovered cluster, there is a good chance that it will be a good center from that cluster.

Let's see whether the same argument applies when sampling uses $\distkm(x, C)^2$ instead of $d(x, C)^2$. If our center set $C$ has a single good center (or a few) from an optimal cluster, say $A$, can we consider the cluster to be covered? No. The probability of sampling from $A$ may remain high since the distance estimate of points in $A$ to the centers in $C \cap A$ may be completely off. When does the distance estimate start becoming tighter? This happens when there are $\Omega(\log{n})$ good centers from $A$ in $C$. This is when we can call the cluster $A$ `settled'. So, instead of `covering' every cluster (if correct distances are known), we care about `settling' every cluster in our current model. So, let us see if we can settle every cluster if we keep sampling centers. We can show that unless the current center set $C$ is already good, there is a good chance that the next center will be sampled from one of the unsettled optimal clusters. Further, we can also argue that conditioned on sampling from an unsettled cluster, there is a good chance that the sampled center will be a good center (i.e., reasonably close to the optimal center). So, as long as we sample sufficiently many centers, every optimal cluster will get settled with high probability. The oversampling factor is $O(\log{n})$, i.e., we end up with a center set $C$ with $O(k \log{n})$ centers. Finally, we will use $\distkm(x, C)$ to assign points $x$ to centers and create a weighted point set $C$ (the weight of a point in $C$ is the number of points assigned to it), on which a standard constant approximation algorithm using strong-oracle queries is used to find the final set of $k$ centers. Since the set $C$ has $O(k \log{n})$ points, we will need $O(k^2 \log^2{n})$ strong-oracle queries to find all the interpoint distances to run the constant approximation algorithm as well as the distance estimates. 
Using known techniques, it can be shown that this final center set gives a constant approximation. To drop the assumption that all optimal clusters have $\Omega(\log{n})$ points, we argue if an optimal cluster does not have adequate points, the oversampling procedure will sample {\em all} the points from that cluster with high probability, which is also a favourable case. More discussions on algorithmic ideas for the $k$-center algorithm and related results can be found in Appendix \ref{sec:detailed-technical-overview}.

\subsection{Related Works} 
\cite{BDJW2024} initiate the study of clustering problems in the weak-strong oracle model and design constant approximation algorithms for $k$-means and $k$-center problems. They also prove that any constant factor approximation algorithm for $k$-means or the $k$-center problem requires $\Omega(k^2)$ strong-oracle queries in the weak-strong model. \cite{GRS2024} obtain the following results for clustering problems in a related model. They design constant factor approximation algorithms for $k$-means and $k$-center clustering problems given access to a {\em quadruplet oracle} that takes two pairs $(x_1,y_1)$ and $(x_2,y_2)$ of points as input and returns whether the pairwise distances are similar or far from each other. A detailed discussion on related work can be found in Appendix \ref{detailed-related-works}.


\section{ALGORITHM FOR $k$-MEANS IN WEAK-STRONG ORACLE MODEL}\label{sec:k-means-weak-strong-main}

In this section we design a constant factor bi-criteria approximate solution for $k$-means in the weak-strong oracle model with the assumption that all optimal clusters have size at least ${480\log n}/{\eps}$. We remove this assumption later. We prove the following result. 

\begin{thm} Let $\eps \in (0,1)$, $\delta\leq 1/3$. There exists a randomized algorithm for $k$-means that makes $O(\frac{k^2 \log^2 n}{\eps^6})$ strong-oracle queries and $O(\frac{nk\log n}{\eps^3})$ weak-oracle queries to give a $\left(O(\frac{\log n}{\eps^3}),40(1+\eps)\right)$ bi-criteria approximation for $k$-means and succeeds with constant probability, assuming every optimal cluster has size at least $\frac{480\log n}{\eps}$. \end{thm}


\begin{algorithm}[htb!]
\DontPrintSemicolon
\caption{Algorithm for $k$-means in weak-strong oracle model}
\label{alg1-main}
\SetKwInOut{Input}{Input}
\SetKwInOut{Output}{Output}
\Input{Dataset $X$ and an integer $k>0$, $\eps\in (0,1)$, $\delta=1/3$.}
\Output{A set of $O(\frac{k \log n}{\eps^3})$ centers and an assignment of $x\in X$ to centers.}

Let $C_1$ be a set of $180\log n$ points chosen arbitrarily from $X$.\hspace{0.2in}\tcc*[f]{\textcolor{blue}{Initial centers}}\;
Set $t = \frac{4320 \cdot 29160}{\eps^3} \cdot k \log n$.\;
\For{$i = 1$ \KwTo $t$}{
    For each $x\in X$, compute $\distkm(x, C_i)$ following Definition \ref{def:km-dist-main}.\;
    Compute distribution $\Tds$ that samples $x\in X$ with probability proportional to $\distkm(x,C_i)^2$.\;
    Sample a point $s_i \in X$ using distribution $\Tds$.\hspace{0.6in}\tcc*[f]{\textcolor{blue}{Sample new centers following $\Tds$}}\;
    Make strong-oracle queries $\So(s_i,c)$ for all $c\in C_i$.\;
    Update $C_{i+1} \gets C_i \cup \{s_i\}$.\;
}

Let $\{c_1, c_2, \ldots, c_h\}$ be an arbitrary ordering of the centers in $C_{t+1}$, where $h=180\log n+t$.

Initialize weights $w(c_i) = 0$ for $i\in [h]$.

\For{$x \in X$}{
    Compute $\distkm(x,C_{t+1})$ and find $c_x$ for which minimum is achieved in Definition \ref{def:km-dist-main}.\;
    Assign $x$ to $c_x$.\;
    Update $w(c_x) \gets w(c_x)+1$.\hspace{0.9in}\tcc*[f]{\textcolor{blue}{Construct weighted instance}}\;
}
\Return{$C_{t+1}$ and assignment of points in $x\in X$ to centers in $C_{t+1}$ as determined above}
\end{algorithm}

Since the weak-oracle gives a wrong answer to a query independently with probability $\delta<1/2$, we define an alternate distance measure between points and a set of centers that we use in our algorithm. Let $C$ denote a set of centers. For any $x\in X$ and any $c\in C$, we query for the distance between $x$ and $c$ to the weak oracle $\Wot$ to obtain a possibly wrong answer $\Wot(x,c)$. For any fixed $x$, we query the weak-oracle for each center $c$ in $C$, and use these query answers $\Wot(x,c)$ to come up with an upper bound on the distance between a point $x$ and a set $C$ of centers defined as follows. We use $\delta \leq 1/3$ for this discussion.

\begin{defn}\label{def:km_dist_twopoints-main} Let $x$ and $y$ be any two points in $X$ and we want to estimate the distance between $x$ and $y$. Let radius $r_y\geq 0$ be such that the ball $B(y,r_y)$ contains at least $180\log n$ points. Then, for any $x\in X$, we define the distance from $x$ to $y$ as $\distkm(x,y)=median\{\Wot(x,z)|z\in B(y,r_y)\}$. 
\end{defn}
We make the following claim with respect to the above distance measure. A similar lemma was proved in \cite{BDJW2024}. The proof can be found in the Appendix.

\begin{lem}\label{lem:km_dist_twopoints-main} Following Definition \ref{def:km_dist_twopoints-main}, for any $x\in X$, with probability at least $\left(1-{1}/{n^{5}} \right)$, we have $|\distkm(x,y)-d(x,y)|\leq r_y$. \end{lem}

We use the above definition to come up with an upper bound on the distance between a point $x$ and a set $C$ of centers. For any $c\in C$, let $r_c$ denote a radius such that $B(c,r_c)$ has at least $180\log n$ points inside it. We use Definition \ref{def:km_dist_twopoints-main} to come up with an upper bound on the distance between $x$ and any $c\in C$ as $\distkm(x,c)+r_c$. Following Lemma \ref{lem:km_dist_twopoints-main}, with high probability, this upper bound holds. That is, with high probability, $d(x,c)\leq \distkm(x,c)+r_c$. Using an union bound over all $c\in C$, all these upper bounds hold and hence, the minimum of these upper bounds would also hold. This motivates the following definition on the distance between a point $x$ and a set $C$ of centers.

\begin{defn}\label{def:km-dist-main} Let $x$ be any point in $X$ and let $C$ be a set of centers of size at least $180\log n$. For any $c\in C$, let $r_c$ denote a radius value such that the ball $B(c,r_c)$ contains at least $180\log n$ points. We define the distance between $x$ and $C$ as $\distkm(x,C)=\min_{c\in C} \{\distkm(x,c)+r_c\}$.
\end{defn}
We state the following lemma with respect to the above distance measure and prove it in Appendix.
\begin{lem}\label{lem:km-dist-main} Let $x$ be any point in $X$ and let $C$ denote a set of centers of size at least $180\log n$. Then, with probability at least $(1-1/n^{4})$, $d(x,C)\leq \distkm(x,C)$. In other words, with probability at least $(1-1/n^{4})$, there exists a center $c\in C$ within a distance of $\distkm(x,C)$ from $x$. \end{lem}


\paragraph{Estimating $\distkm(x,C)$ in weak-strong oracle model} Next, we describe how we estimate $\distkm(x,C)$ for all $x\in X$ in the weak-strong oracle model. We assume that there are at least $180\log n$ centers in $C$. Consider any $c\in C$. We query the strong-oracle $\So$ with all pairs $c_i,c_j\in C$ to obtain the exact distances between centers $c_i$ and $c_j$ as $\So(c_i,c_j)$. For each $c\in C$, we find the smallest $r_c$ such that $B(c,r_c)$ contains at least $180\log n$ points in $C$. Consider any $x\in X$ for which we want to estimate $\distkm(x,C)$. We query the weak-oracle with $x$ and $y\in B(c,r_c)$ to obtain $\Wot(x,y)$ for all $y\in B(c,r_c)$, and use these weak-oracle query answers to compute $\distkm(x,c)$ following Definition \ref{def:km_dist_twopoints-main}. We compute $\distkm(x,c)$ for all $c\in C$. Finally, to compute $\distkm(x,C)$, we find the minimum over all $c\in C$ of $\distkm(x,c)+r_c$, as mentioned in Definition \ref{def:km-dist-main}.

We give a high-level description of Algorithm \ref{alg1-main}. Algorithm \ref{alg1-main} starts with a set $C$ of $180\log n$ centers chosen from $X$, and computes $\distkm(x,C)$ for all $x$. It constructs the $\Tds$-sampling distribution using which a point $x\in X$ is sampled with probability proportional to $\distkm(x,C)^2$. Algorithm \ref{alg1-main} samples a center in each of the $t$ iterations and updates the center set $C$. After $t$ iterations of sampling, it constructs a weighted instance. 
The solution for $k$-means is obtained using a constant approximation algorithm on this weighted instance. The detailed analysis of the algorithm is given in Appendix \ref{sec:k-means-weak-strong}.



\newcommand{\B}{\mathcal{B}}

\section{ALGORITHM FOR $k$-CENTER IN WEAK-STRONG ORACLE MODEL}\label{k-center-weak-strong-main}

\begin{algorithm}[htb!]
\DontPrintSemicolon
\caption{\tt Weak-Greedy Ball Carving}\label{weak-greedy-ball-main} 
\SetKwInOut{Input}{Input}
\SetKwInOut{Output}{Output} 
\Input{Set of points $S$, radius $Rad$} 
\Output{Set $C = \{c_1, \ldots, c_m\}$ and assignment of points in $S$ to centers in $C$}
Initialize: $C = \{\}$\\
\While{$S$ is not empty}{
    If ($|C| = k$) {\bf abort}\\
    If ($|S| \leq 180 k \log{n}$), use strong-oracle queries to find remaining centers covering all points in $S$ within distance $2 Rad$. If this is not possible, {\bf abort}. Otherwise, output $k$ centers.\\
    Pick a subset $T \subset S$ of size $180 k \log{n}$ uniformly at random\\ 
    Query the strong-oracle to find distances between elements in $T$ and use these to find a center $c \in T$ such that $|S_c|$ is maximised, where $S_c \equiv |\B(c, 2Rad) \cap T|$.\hspace{0.1in} \tcc*[f]{\textcolor{blue}{Large sample from cluster}}\;
    If ($|S_c| < 180 \log{n}$) {\bf abort}\\
    Retain $180 \log{n}$ points in $S_c$ ({\it i.e., remove the extra point})\\
    Assign any point $s\in S$ to center $c$ if $\distkm(s,c) \leq 4 Rad$\hspace{0.1in} \tcc*[f]{\textcolor{blue}{Recover cluster}}\; 
    Add $c$ to $C$ and remove all points assigned to $c$ from $S$.       
}
\Return{$C$ and the assignments}
\end{algorithm}

We design a randomized algorithm in the weak-strong oracle model that gives a $6(1+\veps)$-approximation for the $k$-center problem. We sketch the main idea of the algorithm in a simpler setting where accurate distances are known. We assume that the optimal $k$-center radius $r_{opt}$ is known. This is not an unreasonable assumption since we can guess the optimal radius within a $(1\pm \veps)$-factor by iterating over discrete choices of the radius. The algorithm is a simple ``{\em greedy ball-carving}'' procedure and can be stated as follows (also stated as Algorithm 3 in \cite{BDJW2024}): {\em While all points are not removed, pick a center $c$ and remove (carve-out) all points within $2r_{opt}$ of $c$.} We can show that this procedure outputs at most $k$ centers that is a 2-approximate solution to the $k$-center problem.

Our $6(1+\veps)$-approximation algorithm follows the above template. The only complication is that during the ball carving step, where we remove all points within the radius $2r_{opt}$, we need accurate distances, which we do not have in the weak-strong oracle model. Let us continue assuming that the optimal radius $r_{opt}$ is known. We will give a $6$-approximation under this assumption, which changes to $6(1 + \veps)$ when the assumption is dropped. To enable distance estimates, along with a center $c_i$, we must also pick a set of nearby points to $c_i$, which will help in estimating the distance of any point $y$ to $c_i$ following Lemma~\ref{lem:km_dist_twopoints-main}. A reasonable way to do this would be to uniformly sample a set $T_i$ of points (instead of one), use the strong-oracle queries on the subset and pick a center $c_i \in T_i$ with the largest number of points in $T_i$ within a radius $2 r_{opt}$. Let $S_{c_i}$ denote the subset of points in $T_i$ within a distance of $2r_{opt}$ to $c_i$. We can now obtain distance estimates using the tuple $(c_i, S_{c_i})$ using the weak-oracle queries as in Lemma~\ref{lem:km_dist_twopoints-main} and carve out those points and assign them to $c_i$ that satisfy $\distkm(c_i,z) \leq r$ for an appropriate value of $r$. Let $c_i$ belong to the optimal cluster $X_i^*$. What should the appropriate value of $r$ be such that the points in $S \cap X_i^*$ are guaranteed to get carved out and assigned to $c_i$? Since the distance estimates are inaccurate, we must use $r = 4 r_{opt}$. However, this may cause a point that is $6 r_{opt}$ away from $c_i$ to get carved out and assigned to $c_i$. This is the reason we obtain an approximation guarantee of $6$. This high-level analysis lacks two relevant details: (i) how do we eliminate the assumption that $r_{opt}$ is known, and (ii) probability analysis for all points being assigned a center within distance $6 r_{opt}$. We remove the assumption about $r_{opt}$ being known by iterating over discrete choices $Rad = (1+\veps)^0, (1+\veps), (1+\veps)^2, \ldots$ and we argue that the ball carving succeeds for $Rad = r_{opt} (1 + \veps)$ with high probability. Hence, the approximation guarantee becomes $6(1 + \veps)$. Instead of iterating over the possible choices of $Rad$ linearly, we can do that using binary search, which results in the running time and query complexity getting multiplied by a factor of $O\left(\log{{\log{\Delta}}/{\veps}}\right)$, where $\Delta$ is the aspect ratio. Assuming $\Delta$ to be polynomially bounded, this factor becomes $O\left(\log{{\log{n}}/{\veps}}\right)$. We describe our $k$-center algorithm as Algorithm~\ref{weak-greedy-ball-main} and state the theorem that we prove in Appendix \ref{sec:k-center-weak-strong}.

\begin{thm} There exists a $6(1+\eps)$-approximation algorithm for the $k$-center problem that makes $O \left(k^3 \log^2{n} \log{\frac{\log{n}}{\veps}} \right)$ strong-oracle and $O\left(nk\log{n} \log{\frac{\log{n}}{\veps}} \right)$ weak-oracle queries and succeeds with probability at least $(1-1/n^4)^2$. \end{thm}

\section{EXPERIMENTAL RESULTS}\label{sec:experiments-main}

We run our algorithms for $k$-means and $k$-center problems on synthetic as well as real-world datasets to demonstrate that our algorithms can be implemented efficiently in practice, and provide better results compared to \cite{BDJW2024}. Here, we provide experimental results only for $k$-means on real-world dataset MNIST. The experiments were conducted on a server with 1.5 TB RAM and 64 CPU cores. More details about the experiments including results for $k$-means on synthetic datasets and experimental results for $k$-center are given in Appendix \ref{sec:detailed-experiments}.




\paragraph{Datasets} We run our algorithms on synthetic as well as real-world datasets. We generate synthetic data using the Stochastic Block Model (SBM) \citep{MNS2015, A2018}. For the real-world data, we use the MNIST dataset \citep{D2012,MH2008}. The MNIST dataset has ten clusters, $k=10$. The dataset has $n=60k$ images, each with a dimension of $28 \times 28$. Similar to \cite{BDJW2024}, we apply SVD and t-SNE embeddings to embed the dataset into $d=50$ dimensions and $d=2$ dimensions, respectively.


\paragraph{Construction of the weak-oracle} We create perturbed distance matrices $M$ for the datasets, which are used by the weak-oracles to answer queries. These matrices are initialized with the actual distances between points in the datasets. For the MNIST dataset, if two points belong to the same cluster (i.e., the same digit class), we replace the distance between them with a randomly chosen inter-cluster distance (i.e., the distance between two points from different clusters). If the points belong to different clusters (i.e., different digit classes), we assign them a randomly chosen intra-cluster distance.


\paragraph{Baselines in our experiments} Running $k$-means++ with the strong-oracle distances gives us the {\em strong-baseline} for our experiments. We calculate the approximation factors by computing the ratio of the cost of the algorithm with this strong baseline.


\paragraph{Experimental results on $k$-means clustering} We run our $k$-means algorithm on the MNIST dataset using both SVD and t-SNE embeddings. For each embedding, the failure probability $\delta$ of the weak-oracle is set to be $0.1$, $0.2$, and $0.3$. For each of these settings, we create a perturbed distance matrix which the weak-oracle uses to answer queries. Then, for each setting, we run our algorithm on the datasets with varying numbers of strong-oracle queries.

Table \ref{tab:mnist-k-means-main} presents our experimental results for $k$-means clustering on the MNIST dataset using SVD and t-SNE embeddings. We conducted experiments for different combinations of $\delta$ and the number of strong-oracle queries, and in Table \ref{tab:mnist-k-means-main}, we report the values for which the product $(\text{number of strong-oracle queries} \times \log(\text{cost of clustering}))$ is minimized. Table \ref{tab:mnist-k-means-comparison-main} compares the percentage of strong-oracle queries used by our algorithm and by \cite{BDJW2024}. Variance figures are reported in Appendix \ref{sec:detailed-experiments}.

\begin{table}[htbp]
\centering
\normalsize
\caption{Performance of our $k$-means algorithm on MNIST dataset}
\label{tab:mnist-k-means-main}

\begin{tabular}{|c|c|c|c|c|c|c|}
    \hline
    MNIST & \multicolumn{3}{c|}{$\%$ of strong-oracle queries} 
           & \multicolumn{3}{c|}{Approximation factor} \\ \hline
           & $\delta=0.1$ & $\delta=0.2$ & $\delta=0.3$ 
           & $\delta=0.1$ & $\delta=0.2$ & $\delta=0.3$ \\ \hline
    SVD    & $0.00013$ & $0.00019$ & $0.00024$ 
           & $1.019$  & $1.030$  & $1.027$ \\ \hline
    t-SNE  & $0.00054$ & $0.00079$ & $0.00089$ 
           & $1.1595$ & $1.1433$ & $1.0906$ \\ \hline
\end{tabular}

\end{table}

\begin{table}[htbp]
    \centering
    \normalsize
    \caption{Comparison of strong-oracle queries for the $k$-means algorithm on MNIST with \cite{BDJW2024}}
    \label{tab:mnist-k-means-comparison-main}
        \begin{tabular}{|c|c|c|c|c|c|c|}
            \hline
            MNIST & \multicolumn{3}{c|}{\makecell{$\%$ of strong-oracle \\ queries 
            \citep{BDJW2024}}} 
               & \multicolumn{3}{c|}{\makecell{$\%$ of strong-oracle \\ queries (Ours)}} \\ \hline
               & $\delta = 0.1$ & $\delta = 0.2$ & $\delta = 0.3$ 
               & $\delta = 0.1$ & $\delta = 0.2$ & $\delta = 0.3$ \\ \hline
            SVD   & $0.00062$ & $0.00096$ & $0.00063$ & $0.00013$ & $0.00019$ & $0.00024$ \\\hline
            t-SNE & $0.20969$  & $0.20877$ & $0.43814$ & $0.00054$ & $0.00079$ & $0.00089$ \\\hline
        \end{tabular}
    
\end{table}

\noindent \textbf{Comparison with \cite{BDJW2024}}: Table \ref{tab:mnist-k-means-comparison-main} shows that our algorithm requires substantially fewer strong-oracle queries compared to \cite{BDJW2024} for similar approximation guarantees. For MNIST with SVD embedding, the query complexity of our algorithm is at least $61\%$ lower for all values of $\delta$, and for MNIST with t-SNE embedding, the number of queries is at least $99\%$ lower for all $\delta$.

\section*{Acknowledgements}

We thank the anonymous reviewers for various suggestions to improve the quality of the paper, and in particular to improve the constant factors in the analysis of $k$-means and $k$-center algorithms. We have incorporated these suggestions in the appendix. R. Jaiswal acknowledges the support from the SERB, MATRICS grant. R. Jaiswal would like to thank IIT Delhi - Abu Dhabi for hosting him during the later stages of this work. Anup Bhattacharya acknowledges the support from the SERB, CRG grant (CRG/2023/002119).

\bibliographystyle{plainnat}
\bibliography{ref}
\newpage

\clearpage
\appendix
\thispagestyle{empty}


\section{USEFUL LEMMAS}

\begin{lem}(Chernoff bounds) \label{lem:chernoff}(\cite{MU2017}) Let $X_1,\ldots,X_n\in \{0,1\}$ be independent random variables for which $\Pr[X_i=1]=p$. Let $X=\sum_{i=1}^n X_i$ and $\mu=\EE[X]$. Then, the following bounds hold.
\begin{itemize}
    \item For any $0<\gamma\leq 1$, $\Pr[X\geq (1+\gamma)\mu]\leq e^{-\mu \delta^2/3}$.
    \item For any $0<\gamma<1$, $\Pr[X\leq (1-\gamma)\mu]\leq e^{-\mu \delta^2/2}$.
\end{itemize} \end{lem}
\begin{fact}[Biased Cauchy–Schwarz Inequality]\label{fact:bias-cauchy} \cite[Appendix B]{PDE10}
For any $\nu > 0$ and real numbers $a$ and $b$,
   $$(a+b)^2 \le (1+\nu)(a^2/\nu+b^2).$$
\end{fact}

\section{DETAILED RELATED WORKS} \label{detailed-related-works}

Clustering problems such as $k$-means and $k$-center are known to be $\mathsf{NP}$-hard and constant factor approximation algorithms are known for these problems \citep{D2008,G1986,ANSW2019}. \cite{MP2004} prove $\Omega(nk)$ lower bound on the running time of any constant factor randomized approximation of $k$-means. $k$-means++ is a $D^2$-sampling based algorithm that gives an expected $O(\log k)$-approximation for $k$-means. \cite{ADK2009,W2016} show how to obtain a bi-criteria $(2^{2z},O(k))$-approximation guarantee for the $(k,z)$-clustering problem. 

Clustering problems have also been studied in semi-supervised settings. \cite{LM2022} give tight query complexity bounds for the cluster recovery problem with a membership-query oracle, that on a query with any two points, answers whether the points belong to the same cluster or not. When the goal is to obtain is a good approximate clustering solution, a number of works \citep{ABJK2018,MS2017} study clustering problems assuming additional access to a noisy membership-query oracle, also known as the same-cluster query oracle. \cite{GLNS2022} study algorithms for $k$-means where noisy labels are provided for each input point using an adversarial or random perturbation of the labels. \cite{BDJW2024} initiate the study of clustering problems in the weak-strong oracle model and design constant approximation algorithms for $k$-means and $k$-center problems. They also prove that any constant factor approximation algorithm for $k$-means or the $k$-center problem requires $\Omega(k^2)$ strong-oracle queries in the weak-strong model. 

Several other noisy oracle models have been studied in the context of clustering. \cite{ASS2021} introduce two such query models: the adversarial quadruplet oracle and the probabilistic quadruplet oracle.
The adversarial quadruplet oracle takes two pairs of points $(x_1,y_1)$ and $(x_2,y_2)$, as input. If the distances between the two pairs are within a constant factor of each other, the oracle returns an arbitrary answer. Otherwise, it correctly identifies which pair has the smaller distance. The probabilistic quadruplet oracle also takes two pairs of points as input. With probability $p$, it identifies the pair with the smaller distance, and with probability $1 - p$, it returns a wrong answer. \cite{ASS2021} study the $k$-center clustering problem using access to these noisy oracles along with access to a strong-oracle, and designed a constant factor approximation algorithm. Later, \cite{GRS2024} design constant-factor approximation algorithms for both the $k$-means and $k$-center problems in the same model, using $\Tilde{O}(k^2)$ strong-oracle queries.
Recently, \cite{DGLR2025} improve these results by giving constant factor approximation algorithms for both these problems using only $\Tilde{O}(1)$ strong-oracle queries.

\section{DETAILED TECHNICAL OVERVIEW} \label{sec:detailed-technical-overview}

In this section, we describe the main ideas of our algorithms for $k$-means and $k$-center problems in the weak-strong oracle model. One of the most widely used algorithms for $k$-means is the $k$-means++ seeding algorithm that works in $k$ iterations. In the first iteration, it chooses a point uniformly at random. Each of the remaining $(k-1)$ iterations samples a point following a non-uniform and adaptive sampling distribution known as the $D^2$ distribution that depends on the previously chosen centers. Sampling a point from the $D^2$-distribution is called $D^2$-sampling.
The $k$ points sampled using $D^2$-sampling are returned as the solution for $k$-means. 
The $D^2$-distribution is such that the probability of choosing a point $x$ is proportional to the squared distance $d(x, C)^2$, where $C$ is a set of chosen centers and $d(x, C)=\min_{c\in C} d(x,c)$. \cite{ADK2009} showed a constant factor bi-criteria approximation for the oversampling version of $k$-means++ in which $O(k)$ iterations of $D^2$-sampling are used.

One of the main contributions of this work is to show that one can adapt $k$-means++ to the weak-strong oracle model. It is easy to observe that one can simulate $k$-means++ in the strong-oracle model using $O(nk)$ strong-oracle queries. However, things become complicated in the weak-strong model, in which the distances returned by the weak-oracle cannot be trusted. To construct the $D^2$-sampling distribution, one must calculate $d(x, C)$ for any point $x$ and any set $C$ of centers. Since $d(x, C)=\min_{c\in C} d(x,c)$, if we want to compute this distribution exactly, it would require $\Omega(nk)$ strong-oracle queries. Since we don't want to use too many strong-oracle queries, we proceed as follows. Using weak-oracle query answers, we compute an estimate $\distkm(x, C)$ such that with high probability, $d(x,C)$ and $\distkm(x,C)$ are within some small additive factor. We construct the sampling distribution to be used by the algorithm using these $\distkm(x, C)$ values for all points $x$.

Let us see how we compute $\distkm(x,C)$ for all $x\in X$. The main idea here is to exploit the fact that the weak-oracle distances are wrong independently with probability $\delta<1/2$. Consider any $c\in C$. Let us obtain an estimate on $d(x,c)$ using only weak-oracle queries. Suppose there are many points in a ball $B(c,r_c)$ around center $c$ of radius $r_c$ for a reasonably small value of $r_c$. More specifically, suppose $r_c$ is such that $B(c,r_c)$ contains $\Omega(\log n)$ points. 
Then, the key idea is that the median of the weak-oracle distance queries on pairs $(x, y)$ with $y\in B(c,r_c)$ is a good estimate on $d(x, c)$. More formally, we use $\distkm(x,c)=median\{\Wot(x,y):y\in B(c,r_c)\}$ as an estimate for $d(x, c)$. We note that \cite{BDJW2024} also uses the median of query answers to estimate $d(x,c)$. However, their algorithm used these values differently, as we will see in the remaining discussion. 
Since each weak-oracle query answer is wrong independently with probability at most $1/2$, using Chernoff bounds, one can show that for sufficiently large-sized ball $B(c,r_c)$, with high probability, $|d(x,c)-\distkm(x,c)|\leq r_c$. Now, to approximate $d(x,C)$, we compute $\distkm(x,c)+r_c$ for all $c\in C$ and set $\distkm(x,C)=\min_{c\in C} \{\distkm(x,c) + r_c\}$. Once we have the sampling distribution $\Tds$ for which a point $x$ is sampled with probability proportional to $\distkm(x,C)^2$, we show that the approach of \cite{ADK2009} can be adapted to obtain a constant factor bi-criteria approximation for $k$-means that succeeds with high probability. Even though the analysis becomes non-trivial and we incur an approximation loss, the high-level analysis ideas of \cite{ADK2009}  goes through.

The main idea is to run the oversampling $k$-means++ where we iteratively sample and update the center set $C$. However, the sampling is with respect to $\distkm(x, C)^2$ instead of $d(x, C)^2$. The goal is to pick a good center from every optimal cluster. Let us see why oversampling using $\distkm(x, C)^2$ will be sufficient to find good centers from every optimal cluster, and hence obtain a good approximation guarantee. For the simplicity of discussion, let us first assume that every optimal cluster has $\Omega(\log{n})$ points. We will later see how to drop this assumption. Let us draw a parallel with oversampling $k$-means++ that samples using accurate distance values. Why does it manage to pick good centers from every optimal cluster? Consider an intermediate center set $C$. Certain optimal clusters will have a good representative center in $C$, whereas others remain `uncovered'. Sampling using the $D^2$ distribution boosts the probability of sampling from an uncovered cluster, so there is a good chance that the next sample belongs to an uncovered cluster.
Moreover, we can argue that given the next sample is from an uncovered cluster, there is a good chance that it will be a good center from that cluster. Let's see whether the same argument applies when sampling uses $\distkm(x, C)^2$ instead of $d(x, C)^2$. If our center set $C$ has a single good center (or a few) from an optimal cluster, say $A$, can we consider the cluster to be covered? No. The probability of sampling from $A$ may remain high since the distance estimate of points in $A$ to the centers in $C \cap A$ may be completely off. When does the distance estimate start becoming tighter? This happens when there are $\Omega(\log{n})$ good centers from $A$ in $C$. This is when we can call the cluster $A$ `settled'. So, instead of `covering' every cluster (if correct distances are known), we care about `settling' every cluster in our current model. So, let us see if we can settle every cluster if we keep sampling centers. We can show that unless the current center set $C$ is already good, there is a good chance that the next center will be sampled from one of the unsettled optimal clusters. Further, we can also argue that conditioned on sampling from an unsettled cluster, there is a good chance that the sampled center will be a good center (i.e., reasonably close to the optimal center). So, as long as we sample sufficiently many centers, every optimal cluster will get settled with high probability. The oversampling factor is $O(\log{n})$, i.e., we end up with a centre set $C$ with $O(k \log{n})$ centers. Finally, we will use $\distkm(x, C)$ to assign points $x$ to centers and create a weighted point set $C$ (the weight of a point in $C$ is the number of points assigned to it), on which a standard constant approximation algorithm using strong-oracle queries is used to find the final set of $k$ centers. Since the set $C$ has $O(k \log{n})$ points, we will need $O(k^2 \log^2{n})$ strong oracle queries to find all the interpoint distances to run the constant approximation algorithm as well as the distance estimates. 
Using known techniques, it can be shown that this final center set gives a constant approximation. To drop the assumption that all optimal clusters have $\Omega(\log{n})$ points, we argue if an optimal cluster does not have adequate points, the oversampling procedure will sample {\em all} the points from that cluster with high probability, which is also a favourable case. 



For the $k$-center problem, we design a simple ball carving algorithm that gives a $6(1 + \veps)$-approximate solution. For the current discussion, we will make the simplifying assumptions that the optimal radius $R$ is known. We will later see how to drop these assumptions. We will design a ball carving procedure that runs in $k$ iterations and is guaranteed to carve out all the points of a new optimal cluster in every round. Carving out a ball means picking a center $c$ and removing all points within a ball of appropriate radius centered at $c$. As we did for the $k$-means problem, we'll draw a parallel with the ball carving procedure that works when distances are accurately known. In the perfect case, suppose we pick an arbitrary point as center in each round and carve out the ball of radius $2R$ in every round. For any optimal cluster $A$ and any point $a \in A$, a ball of radius $2R$ will contain all points in $A$. So, in every round, the center $c$ is guaranteed to belong to a previous uncarved ball and we are guaranteed to carve out all points in the optimal cluster to which $c$ belongs. So, in $k$ iterations, we will carve out all points within radius $2R$. This gives us a $2$-approximate solution.
Let's see if the same argument applies to our current weak-strong oracle model.
The main issue while carving a ball centered at the next chosen center $c$ is that we may not have accurate distance estimates. What should be the appropriate center $c$ and the appropriate radius $r$ so that if we carve out all points whose distance estimate to $c$ is within $r$, then all points from the optimal cluster to which $c$ belongs get removed? 
Suppose all points from $i-1$ optimal clusters have been carved out in the first $i-1$ iterations. Let's focus on the $i^{th}$ iteration. If only a few points are remaining, then we can find the centers using strong queries for all remaining points. Let us assume that at least $\Omega(k \log{n})$ points remain.
We can argue that if we randomly sample a set $T$ of $\Omega(k \log{n})$ points from the remaining set of points, then there is a point $c \in T$ such that there are $\Omega(\log{n})$ points in $B(c, 2R) \cap T$. We can use strong-oracle queries on the set $T$ to locate center $c$ and $\Omega(\log{n})$ points from $T$ within $2R$ of $c$.
Let $c$ belong to the optimal cluster $A$. 
So, if we consider the closest $O(\log{n})$ points to $c$ when calculating the median distance (this is the same as we did for $k$-means), then our distance estimate is off by an additive factor of $2R$. This means that if we carve out all points within the estimated distance $4R$ from $c$, it is guaranteed to remove all points in $A$. However, the removed points may also include points with true distance $6R$ away from $c$, which gets assigned to $c$. So, we will end with an approximation guarantee of $6$. What is the number of strong oracle queries required for the procedure? In each round, we need $O(k^2 \log^2{n})$ strong queries and there are $k$ rounds. So, the total number of strong queries is $O(k^3 \log^2{n})$. How do we remove the assumption that the optimal radius $R$ is known? This is done by guessing the radius using binary search in an appropriate range and out of a discrete set of possibilities determined by the error parameter $\veps > 0$. This makes the approximation guarantee $6(1+\veps)$ and adds a multiplicative factor of $\log{\frac{\log{n}}{\veps}}$ to the number of strong oracle queries.

\cite{BDJW2024} show a $\Omega(k^2)$ strong-oracle query lower bound in the weak-strong model for any constant factor approximation of $k$-means or $k$-center problems. However, no such weak-oracle query lower bound was known for $k$-means. In this work, we show a $\Omega(\frac{nk}{(1-2\delta)^2})$ weak-oracle query lower bound for any constant approximation of $k$-means in the weak-oracle model. To show this lower bound, we use a known $\Omega(\frac{nk}{(1-2\delta)^2})$ query lower bound of \cite{MS2017} for the cluster recovery problem using a faulty query oracle.

\section{INFORMAL JUSTIFICATION OF THE WEAK-ORACLE MODEL:}\label{sec:just-weak}


In the weak-oracle model, it is assumed that the weak-oracle returns the exact distance between any two points independently with probability $(1-\delta)$ and with the remaining probability, it can report any arbitrary value. We first note that one can relax the requirement that the weak-oracle needs to report the exact distance between two points as follows. The weak-oracle could return a value within a constant factor of the actual distance, with high probability. Our algorithms would work with this modified weak-oracle and will give a solution of quality at most a constant factor worse compared to the earlier solution. Regarding the independence assumption of the weak-oracle, we note that our algorithms should also work in bounded independence models. We apply Chernoff-Hoeffding under independence, which can be replaced with other concentration inequalities under bounded independence.  However, the algorithms do not work if the query answers are arbitrarily correlated. We are not aware of many studies that design algorithms for clustering problems that work with arbitrarily correlated Oracle answers. One recent work along these lines is \cite{BGT2024}, which designs mechanisms for facility location problems given noisy Oracle access (oracle gives noisy locations for the clients) and highlights potential difficulties in obtaining strong algorithmic results without the independence assumption. On the other hand, there are applications where assuming that the errors are independent is not completely unreasonable, such as when distances are obtained using a crowd-sourced platform. So, the independence assumption provides a good initial platform for understanding important questions, such as whether it is possible to retain the approximation guarantees of widely used algorithms such as k-means++ (for k-means) and ball-carving (for k-center) in this model.

One natural attempt to model noisy distances might be where the weak-oracle answers are values sampled from a probability distribution. First, let us consider the scenario in which the oracle answers are not persistent. In this case, one can ask the same query multiple times and obtain concentration bounds for the queried distances. This makes the problem easy. Now, let us assume that the query answers are persistent. If the probability distributions have strong concentration properties, then with high probability (say, with probability $(1-\delta)$), the samples will be close to the expected value and with remaining probability (with probability at most $\delta)$, the sample will be far from the expected value. We believe that this can be modeled in the weak-strong oracle model that we consider in this paper with an appropriate choice of parameters. First, note that one can relax the weak-oracle to require that it only returns an answer within a constant factor of the actual value, and our algorithms with this modified oracle will return solutions that is at most a constant factor worse. The above discussion assumes that the distances are drawn from independent probability distributions and the sampling is also done independently.

\section{CLUSTERING IN STRONG-ORACLE MODEL} \label{sec:strong-oracle}

\paragraph{Strong-oracle model:} We note that strong-oracle model is precisely the well-known classical model. So, let us remind ourselves of what is known. We observe that $k$-means++ algorithm can be executed using $O(nk)$ strong-oracle queries and using known results, one can show that it gives an $O(\log k)$-approximation in expectation for the $k$-means problem \citep{AV2007}. Further, the oversampling version of $k$-means++ gives constant factor bi-criteria approximation for $k$-means using only $O(nk)$ strong-oracle queries \citep{ADK2009}. We also note that a $\Omega(nk)$ strong-oracle query lower bound for any constant factor approximation of $k$-means follows from  \cite{MP2004}. Even though we state results only for the $k$-means problem, it is possible to extend these ideas to obtain a $O(2^{2z})$ bi-criteria approximation in expectation for the $(k,z)$-clustering using $O(nk)$ strong-oracle queries \citep{W2016}. For the $k$-center problem, we observe that $O(nk)$ strong-oracle queries suffice to run Gonzalez's farthest-point algorithm to obtain a $2$-approximation \citep{G1986}. These results are discussed in Appendix \ref{sec:strong-oracle}.

\subsection{Algorithm for \texorpdfstring{$k$}{k}-means Clustering}

\begin{thm} \label{thm:strong_queries} There exists an algorithm that makes $O(nk)$ strong-oracle queries and returns a constant factor bi-criteria approximation for $k$-means with constant probability. \end{thm}

\begin{proof} We show that the oversampling $k$-means++ algorithm, which samples centers using $D^2$-sampling for $O(k)$ iterations, can be simulated using $O(nk)$ strong-oracle queries. Recall that oversampling $k$-means++ starts with a center $c_1$ chosen uniformly at random and initializes a center set $C=c_1$. We query the strong-oracle with $(x,c_1)$ for all $x\in X$ to obtain $SO(x,c_1)$ for all $x\in X$, and construct the $D^2$-sampling distribution. This step requires $O(n)$ strong-oracle queries. In the second iteration, the oversampling $k$-means++ algorithm samples a center $c_2$ following $D^2$-distribution with respect to $C=\{c_1\}$. We update $C=C\cup \{c_2\}$. We again query the strong-oracle for $x\in X$ and $c_2$, and update the sampling distribution using the query results. In this manner for $O(k)$ iterations, we make in total $O(nk)$ strong-oracle queries and using the query results, run the oversampling $k$-means++ algorithm. Following \cite{ADK2009}, oversampling $k$-means++ gives a constant approximation for $k$-means with constant probability. \end{proof}

The above bound on the number of strong-oracle queries is tight for a constant factor approximation of $k$-means, following a known lower bound \cite{MP2004}.

\begin{thm} Any randomized algorithm that provides a constant-factor approximation for the $k$-means problem with constant probability requires $\Omega(nk)$ strong-oracle queries. \end{thm}

\begin{proof} Mettu and Plaxton \citep{MP2004} prove a $\Omega(nk)$ lower bound on the running time of any constant factor randomized approximation of $k$-means. The $\Omega(nk)$ query lower bound result follows from this lower bound of \cite{MP2004}.  \end{proof}

\subsection{Algorithm for \texorpdfstring{$k$}{k}-center Clustering}

\begin{thm}\label{thm:k-center-strong} There exists a $2$-approximation algorithm for the $k$-center problem using $O(nk)$ queries. \end{thm}
\begin{proof} We observe that the farthest-point algorithm of \cite{G1986} can be implemented using only $O(nk)$ strong-oracle queries. This will give a $2$-approximation for the $k$-center problem.\end{proof}

\section{ALGORITHM FOR \texorpdfstring{$k$}{k}-MEANS IN WEAK-STRONG ORACLE MODEL}\label{sec:k-means-weak-strong}

For ease of reading, we reproduce here the entire algorithm for $k$-means clustering and its analysis. In this section, we design a constant factor bi-criteria approximate solution for $k$-means in the weak-strong oracle model with the assumption that all optimal clusters have size at least $\frac{480\log n}{\eps}$. We remove this assumption later. We prove the following result. 

\begin{thm} Let $\eps \in (0,1)$, $\delta\leq 1/3$. There exists a randomized algorithm for $k$-means that makes $O(\frac{k^2 \log^2 n}{\eps^6})$ strong-oracle queries and $O(\frac{nk\log n}{\eps^3})$ weak-oracle queries to give a $\left(O(\frac{\log n}{\eps^3}),40(1+\eps)\right)$ bi-criteria approximation for $k$-means and succeeds with constant probability, assuming every optimal cluster has size at least $\frac{480\log n}{\eps}$. \end{thm}

\subsection{Algorithm and Analysis}

Since the weak-oracle gives a wrong answer to a query independently with probability $\delta<1/2$, we define an alternate distance measure between points and a set of centers that we use in our algorithm. Let $C$ denote a set of centers. For any $x\in X$ and any $c\in C$, we query for the distance between $x$ and $c$ to the weak oracle $\Wot$ to obtain a possibly wrong answer $\Wot(x,c)$. For any fixed $x$, we query the weak-oracle for each center $c$ in $C$, and use these query answers $\Wot(x,c)$ to come up with an upper bound on the distance between a point $x$ and a set $C$ of centers defined as follows. We use $\delta \leq 1/3$ for this discussion.

\begin{defn}\label{def:km_dist_twopoints} Let $x$ and $y$ be any two points in $X$ and we want to estimate the distance between $x$ and $y$. Let radius $r_y\geq 0$ be such that the ball $B(y,r_y)$ contains at least $180\log n$ points. Then, for any $x\in X$, we define the distance from $x$ to $y$ as $\distkm(x,y)=median\{\Wot(x,z)|z\in B(y,r_y)\}$. 
\end{defn}
We make the following claim with respect to the above distance measure. A similar lemma was proved in \cite{BDJW2024}. We add a proof here for completeness.

\begin{lem}\label{lem:km_dist_twopoints} Following Definition \ref{def:km_dist_twopoints}, for any $x\in X$, with probability at least $\left(1-\frac{1}{n^{5}} \right)$, we have $|\distkm(x,y)-d(x,y)|\leq r_y$. \end{lem}
\begin{proof} 
Let us consider the ball $B(y,r_y)$ centered at $y$ and radius $r_y\geq 0$. By definition, $B(y,r_y)$ has at least $l=180\log n$ points inside it. For each point $z\in B(y,r_y)$, the weak-oracle returns a wrong answer for $\Wot(x,z)$ with probability $\delta \leq 1/3$. Hence, the expected number of wrong answers is at most $l/3$. For each query to the weak-oracle with $x$ and $z\in B(y,r_y)$, we have $|\Wot(x,z)-d(x,z)|>r_y$ with probability at most $1/3$. For the median of the query answers $\distkm(x,y)=median\{\Wot(x,z):z\in B(y,r_y)\}$ to satisfy $|\distkm(x,y)-d(x,y)|>r_y$, we need at least half of the query answers $\Wot(x,z)$ to satisfy $|\Wot(x,z)-d(x,z)|>r_y$. Using Lemma \ref{lem:chernoff}, this probability is at most $e^{-1/3 \cdot 1/4 \cdot l/3}\leq 1/n^5$.
\end{proof}

We use the above definition to come up with an upper bound on the distance between a point $x$ and a set $C$ of centers. For any $c\in C$, let $r_c$ denote a radius such that $B(c,r_c)$ has at least $180\log n$ points inside it. We use Definition \ref{def:km_dist_twopoints} to come up with an upper bound on the distance between $x$ and any $c\in C$ as $\distkm(x,c)+r_c$. Following Lemma \ref{lem:km_dist_twopoints}, with high probability, this upper bound holds. That is, with high probability, $d(x,c)\leq \distkm(x,c)+r_c$. Using an union bound over all $c\in C$, all these upper bounds hold and hence, the minimum of these upper bounds would also hold. This motivates the following definition on the distance between a point $x$ and a set $C$ of centers.

\begin{defn}\label{def:km-dist} Let $x$ be any point in $X$ and let $C$ be a set of centers of size at least $180\log n$. For any $c\in C$, let $r_c$ denote a radius value such that the ball $B(c,r_c)$ contains at least $180\log n$ points. We define the distance between $x$ and $C$ as $\distkm(x,C)=\min_{c\in C} \{\distkm(x,c)+r_c\}$.
\end{defn}
We state the following lemma with respect to the above distance measure.
\begin{lem}\label{lem:km-dist} Let $x$ be any point in $X$ and let $C$ denote a set of centers of size at least $180\log n$. Then, with probability at least $(1-1/n^{4})$, $d(x,C)\leq \distkm(x,C)$. In other words, with probability at least $(1-1/n^{4})$, there exists a center $c\in C$ within a distance of $\distkm(x,C)$ from $x$. \end{lem}

\begin{proof} Let $x$ be any point in $X$. Using Lemma \ref{lem:km_dist_twopoints}, for any $c\in C$, we have with probability at least $(1-1/n^{5})$, $\distkm(x,c)+r_c\geq d(x,c)$. Using a union bound over all $c\in C$, with probability at least $(1-1/n^{4})$, all these upper bounds on $d(x,c)$ hold. Hence, $\distkm(x,C)=\min_{c\in C} \{\distkm(x,c)+r_c\}\geq \min_{c\in C} d(x,c)=d(x,C)$, with probability at least $(1-1/n^{4})$. Since with probability at least $(1-1/n^{4})$, these upper bounds $d(x,c)\leq \distkm(x,c)+r_c$ hold for $c\in C$, and the minimum over these upper bounds also hold. This implies that with this probability, there exists a center $c\in C$ within distance $\distkm(x,C)$ from $x$. \end{proof}

\paragraph{Estimating $\distkm(x,C)$ in weak-strong oracle model} Next, we describe how we estimate $\distkm(x,C)$ for all $x\in X$ in the weak-strong oracle model. We assume that there are at least $180\log n$ centers in $C$. Consider any $c\in C$. We query the strong-oracle $\So$ with all pairs $c_i,c_j\in C$ to obtain the exact distances between centers $c_i$ and $c_j$ as $\So(c_i,c_j)$. For each $c\in C$, we find the smallest $r_c$ such that $B(c,r_c)$ contains at least $180\log n$ points in $C$. Consider any $x\in X$ for which we want to estimate $\distkm(x,C)$. We query the weak-oracle with $x$ and $y\in B(c,r_c)$ to obtain $\Wot(x,y)$ for all $y\in B(c,r_c)$, and use these weak-oracle query answers to compute $\distkm(x,c)$ following Definition \ref{def:km_dist_twopoints}. We compute $\distkm(x,c)$ for all $c\in C$. Finally, to compute $\distkm(x,C)$, we find the minimum over all $c\in C$ of $\distkm(x,c)+r_c$, as mentioned in Definition \ref{def:km-dist}.

\begin{algorithm}[H]
\caption{Algorithm for $k$-means in weak-strong oracle model}\label{alg1}
\SetKwInOut{Input}{Input}
\SetKwInOut{Output}{Output}
\Input{Dataset $X$ and an integer $k>0$, $\eps\in (0,1)$, $\delta=1/3$.}
\Output{A set of $O(\frac{k \log n}{\eps^3})$ centers and an assignment of $x\in X$ to centers.}

$\triangleright$ Let $C_1$ be a set of $180\log n$ points chosen arbitrarily from $X$.\;
Set $t = \frac{4320\cdot 29160}{\eps^3} \cdot k \log n$\;

\For{$i = 1$ \KwTo $t$}{
    $\triangleright$ For each $x\in X$, compute $\distkm(x, C_i)$ following Definition \ref{def:km-dist}.\;
    $\triangleright$ Construct distribution $\Tds$ that samples a point $x\in X$ with probability proportional to $\distkm(x,C_i)^2$.\;
    $\triangleright$ Sample a point $s_i \in X$ using distribution $\Tds$.\;
    $\triangleright$ Make a strong-oracle query with $\So(s_i,c)$ for all $c\in C_i$.\;
    $\triangleright$ Update $C_{i+1} \gets C_i \cup \{s_i\}$.
}

$\triangleright$ Let $\{c_1, c_2, \ldots, c_h\}$ be an arbitrary ordering of the centers in $C_{t+1}$, where $h=180\log n+t$.

$\triangleright$ Initialize weights $w(c_i) = 0$ for $i\in [h]$.

\For{$x \in X$}{
    $\triangleright$ Compute $\distkm(x,C_{t+1})$ and determine $c_x$ for which the minimum is achieved in Definition \ref{def:km-dist}.\;
    $\triangleright$ Assign $x$ to $c_x$.\;
    $\triangleright$ Update $w(c_x) \gets w(c_x)+1$.
}

\Return{$C_{t+1}$ and the assignment of each $x \in X$ to a center in $C_{t+1}$.}
\end{algorithm}

Let $C$ be any set of centers of size at least $180\log n$. Let $\Tds$ denote the proportional sampling distribution over points in $X$ with respect to $C$. The probability of sampling a point $x\in X$ following distribution $\Tds$ is proportional to $\distkm(x,C)^2$. For any point $x \in X$, the probability of sampling $x$ with respect to $C$ is given as $\frac{\distkm(x,C)^2}{\sum_{x' \in X}\distkm(x',C)^2}$. Using Lemma \ref{lem:km-dist} with a union bound over all points $x\in X$, we have that with probability at least $(1-\frac{1}{n^{3}})$, the distance estimates $\distkm(x,C)$ given by Definition \ref{def:km-dist} provide valid upper bounds for $d(x,C)$ for all $x\in X$. For the following discussion, we condition on this high probability event $\E$.

\begin{lem}\label{lem:dist_est_cor} Let $C$ be any set of centers of size at least $180\log n$. Let $\E$ denote the event that $\distkm(x,C) \geq d(x, C)$ for all points $x\in X$. Then, $\E$ holds with probability at least $1-1/n^{3}$. \end{lem}

\begin{proof} The proof follows by using Lemma \ref{lem:km-dist} for all $x\in X$ and applying a union bound. \end{proof}

Since we do not know the exact distances for most pairs of points in a cluster, we define below the cost of a cluster using the distance estimates defined in Definition \ref{def:km-dist}.

\begin{defn}(Cost of a cluster) Let $A$ be any optimal cluster and let $C$ denote a set of centers. Then, the cost of cluster $A$ with respect to center set $C$ is given as $\tp(A,C)=\sum_{x\in A} \distkm(x,C)^2$. \end{defn}

We next define the notion of \textit{settled} clusters that give good approximation guarantees with respect to the set $C$ of centers chosen by the algorithm.

\begin{defn}(Settled clusters) Let $C$ be a set of centers. An optimal cluster $A$ is said to be {\it settled} with respect to $C$ if $\tp(A,C)\leq 20(1+\eps) \cdot OPT(A)$. Otherwise, the cluster $A$ is said to be unsettled. \end{defn}

Let $C_{i-1}$ denote the set of centers at the end of iteration $(i-1)$. Let $A$ be an optimal cluster. Let the set of {\it settled} optimal clusters at the start of iteration $i$ be given as $G_i= \{A:\tp(A,C_{i-1})\leq 20(1+\eps) \cdot OPT(A)\}$ and the set of {\it unsettled} optimal clusters be given as $B_i= \{A:\tp(A,C_{i-1})> 20(1+\eps) \cdot OPT(A)\}$. 

\begin{lem}\label{lem:progress} During the $i^{th}$ iteration of Algorithm \ref{alg1}, either $\tp(X,C_{i-1}) \leq 40(1+\eps) \cdot OPT$ or the probability of sampling a center from some unsettled optimal cluster $A\in B_i$ using distribution $\Tds$ is at least 1/2. \end{lem}

\begin{proof} If $\tp(X,C_{i-1})\leq 40(1+\eps) \cdot OPT$, then $C_{i-1}$ gives a $40(1+\eps)$-approximation. For the following, we assume $\tp(X,C_{i-1})> 40(1+\eps) \cdot OPT$. We show the probability of sampling a point from an unsettled optimal cluster $A\in B_i$ is at least $1/2$. Let $x$ be a point sampled using distribution $\Tds$. Then,
\begin{eqnarray*}
    \Pr(x \in A, A \in B_i) &=& \frac{\sum_{A \in B_i}\tp(A,C_{i-1})}{\tp(X,C_{i-1})}\\
    &=& 1 - \frac{\sum_{A \in G_i}\tp(A,C_{i-1})}{\tp(X,C_{i-1})}\\
    &\geq& 1 - 20(1+\eps) \cdot \frac{\sum_{A \in G_i}OPT(A)}{40(1+\eps) \cdot OPT}\\
    &\geq& 1 - 1/2=1/2
\end{eqnarray*}
This completes the proof of the lemma.
\end{proof}

Let $A\in B_i$ be an unsettled optimal cluster. Let $\mu_A$ be the centroid of $A$.
We define $r$ to be the average radius of cluster $A$ given as $r=\sqrt{OPT(A)/|A|}$. We define an inner ring $IR(A,\alpha)$ around $\mu_A$ of radius $\alpha r$, for some $\alpha>1$ to be chosen later. Formally, $IR(A,\alpha)=\{x\in A: ||x-\mu_A||\leq \alpha r\}$.

One of the main ideas in the analysis of Algorithm \ref{alg1} is to show that once we sample at least $180 \log |A|$ centers from the inner ring of an unsettled optimal cluster $A$, the cost of cluster $A$ with respect to the set of centers becomes at most a constant times its optimal cost $OPT(A)$ with high probability. Fix any point $x\in A$. Let $C_i$ denote the set of centers at the end of iteration $i$ such that it contains at least $180\log |A|$ points from the inner ring $IR(A,\alpha)$ of the optimal cluster $A$. Using arguments similar to Lemma \ref{lem:dist_est_cor}, we can show with probability at least $(1-1/|A|^{3})$, there exists a center $c_x$ in $C_i$ within distance $\distkm(x,C_i)$ from $x$ and this holds for all $x\in A$. Since $C_i$ has at least $180\log |A|$ points from the inner ring of $A$, with probability at least $(1-1/|A|^{3})$, the distance $\distkm(x,C_i)$ from $x$ to $c_x$ is at most the distance from $x$ to its farthest point in the inner ring of $A$ plus the diameter of the inner ring. This corresponds to $c_x\in C_i\cap IR(A,\alpha)$ being the farthest point in the inner ring from $x$ and $r_{c_x}\leq 2\alpha r$.

\begin{lem}\label{lem:unsettle-settle} Let $C_{i-1}$ denote the set of sampled centers at the end of iteration $(i-1)$, and let us assume that during iteration $i$, we have sampled a point such that $C_i$ has at least $180\log n$ points from the inner ring $IR(A,\alpha)$ of $A$, where $A$ is an unsettled optimal cluster with respect to $C_{i-1}$. For $\alpha>1$, we have $\tp(A,C_i) \leq 2(1+9\alpha^2) \cdot OPT(A)$ with probability at least $(1-1/|A|^{3})$. \end{lem}

\begin{proof} Following Lemma \ref{lem:dist_est_cor}, with probability at least $(1-1/n^{3})$, for all points $x\in A$, we have $\distkm(x,C_i)$ to be at most $\distkm(x,c_x)+r_{c_x}$, where $c_x\in C_i\cap IR(A,\alpha)$ with radius $r_{c_x}\leq 2\alpha r$. We have the following.
    \begin{align}
        \tp(A,C_i)&= \sum_{x\in A} \distkm(x,C_i)^2 \notag\\
        &\leq \sum_{x\in A} (d(x,c_x)+2\alpha r)^2, ~\text{where}~c_x\in C_i\cap IR(A,\alpha)\notag\\
        &\leq \sum_{x\in A} (d(x,\mu_A)+3\alpha r)^2 \label{eqn:cost-a-ci}\\
        &\leq \sum_{x\in A} [2d^2(x, \mu_A)+18 \alpha^2 r^2] \notag\\
        &\leq 2 OPT(A) + 18\sum_{x \in A}(\alpha r)^2 \notag\\
        &\leq  2 OPT(A) + 18\cdot |A|\cdot \alpha^2\frac{OPT(A)}{|A|} \notag\\
        &=2(1+9\alpha^2) \cdot OPT(A) \notag
  \end{align}
This completes the proof of the lemma.
\end{proof}

Next, we show that the probability of sampling a point following distribution $\Tds$ from the inner ring $IR(A,\alpha)$ of an unsettled optimal cluster $A$ is not too small. First, we show that the inner ring of an unsettled optimal cluster $A$ contains a significant fraction of points of $A$.

\begin{lem}\label{lem:inner-ring} For any unsettled optimal cluster $A$, there are at least $|A|(1-\frac{1}{\alpha^2})$ points in the inner ring $IR(A,\alpha)$ of that cluster, where $\alpha>1$. \end{lem}

\begin{proof}
We have:
\begin{align*}
     OPT(A) &\geq \sum_{x \in A\setminus IR(A,\alpha)}||x-\mu_A||^2\\
     & \geq |A \setminus IR(A,\alpha)|(\alpha r)^2\\
     &= \left( 1 - \frac{|IR(A,\alpha)|}{|A|}\right)|A|(\alpha r)^2\\
     & =\left( 1 - \frac{|IR(A,\alpha)|}{|A|}\right)\alpha^2 OPT(A)
\end{align*}
Rearranging terms, we get $|IR(A,\alpha)| \geq |A|(1 - \frac{1}{\alpha^2})$.
\end{proof}

Lemma \ref{lem:inner-ring} tells us that the inner ring of an unsettled optimal cluster $A$ contains at least $(1-1/\alpha^2)$-fraction points of cluster $A$. Conditioned on sampling a point from an unsettled optimal cluster $A$ following distribution $\Tds$ during the $i^{th}$ iteration, the probability of sampling a point from the inner ring of $A$ is $\frac{\tp(IR(A,\alpha),C_{i-1})}{\tp(A,C_{i-1})}$. To get a lower bound on this probability, we need a lower bound on $\tp(IR(A,\alpha),C_{i-1})$ and an upper bound on $\tp(A,C_{i-1})$. We show the following.

\begin{lem}\label{lem:sample-inner-ring} $\Pr(x \in IR(A,\alpha)\mid x \in A , A \in B_i)\geq (1-1/\alpha^2) \frac{(\sqrt{1+10\eps/9}-\alpha)^2}{180(1+\eps)}$. \end{lem}
\begin{proof} Conditioned on sampling a point from an unsettled optimal cluster $A$, the probability of sampling a point from the inner ring of that cluster during the $i^{th}$ iteration is given as

\begin{align*}
        \Pr(x \in IR(A,\alpha)\mid x \in A ,A \in B_i) &= \frac{\tp(IR(A,\alpha),C_{i-1})}{\tp(A,C_{i-1})}
\end{align*}

Let $c_{\mu_A}$ denote the center in $C_{i-1}$ for which $\distkm(\mu_A,C_{i-1})$ is minimized. Let $v=\distkm(\mu_A,C_{i-1})$. Recall that for any point $x\in A$, $c_x$ denotes the center in $C_{i-1}$ for which $\distkm(x,C_{i-1})$ is minimized.

\begin{claim}\label{claim:rel_med_dists} Let $x$ and $y$ be any two points for which $\distkm(x,C_{i-1})$ and $\distkm(y,C_{i-1})$ are minimized by centers $c_x\in C_{i-1}$ and $c_y\in C_{i-1}$, respectively. Then, assuming event $\E$ holds for all iterations of Algorithm \ref{alg1}, we have $\distkm(y,c_x)\leq d(x,y)+\distkm(x,c_x)+2r_{c_x}$ and $\distkm(x,c_y)\leq d(x,y)+\distkm(y,c_y)+2r_{c_y}$. \end{claim}

\begin{proof} Let $x$ be a point in the optimal cluster $A$. Since $\distkm(x,C_{i-1})$ is minimized by $c_x\in C_{i-1}$, for point $x$, the center $c_x$ minimizes $\min_{c\in C_{i-1}} (\distkm(x,c)+r_c)$ and the ball $B(c_x,r_{c_x})$ has at least $180\log n$ points in it. Conditioning on event $\E$ for all iterations of Algorithm \ref{alg1}, we have for both $\distkm(x,c_x)$ and $\distkm(y,c_x)$, the median distances are realized by points inside the ball $B(c_x,r_{c_x})$. Hence, $\distkm(y,c_x)\leq d(x,y)+\distkm(x,c_x)+2r_{c_x}$. Similarly, the other inequality holds.
\end{proof}
Next, we obtain an upper bound on $\tp(A,C_{i-1})$.
\begin{align}
    \tp(A,C_{i-1})& = \sum_{x \in A}\distkm(x, C_{i-1})^2 \notag \\
        &= \sum_{x\in A} (\distkm(x,c_x)+r_{c_x})^2 \notag\\
        &\leq \sum_{x\in A} (\distkm(x,c_{\mu_A})+r_{c_{\mu_A}})^2 \notag\\
        &\leq \sum_{x\in A} (d(x,\mu_A)+\distkm(\mu_A,c_{\mu_A})+3r_{c_{\mu_A}})^2~~\text{(using Claim \ref{claim:rel_med_dists})} \notag\\
        &\leq \sum_{x\in A} (d(x,\mu_A)+3\distkm(\mu_A,C_{i-1}))^2 \label{eqn:lower-v-1}\\
        &\leq \sum_{x\in A} (2d(x,\mu_A)^2+ 18\distkm(\mu_A,C_{i-1})^2)\notag\\
        &= \sum_{x\in A} (2d(x,\mu_A)^2+ 18 v^2)\notag\\
        &= 2 OPT(A)+ 18 |A|v^2 \notag\\
        &= 2|A|(r^2+9 v^2)\notag
\end{align}

Now, we compute a lower bound on $\tp(IR(A,\alpha),C_{i-1})$. Let $b$ be any point in the inner ring $IR(A,\alpha)$ of optimal cluster $A$ and let $c_b$ be the center in $C_{i-1}$ for which $\distkm(b,C_{i-1})$ is minimized. We observe that $\distkm(\mu_A,C_{i-1})\leq \distkm(\mu_A,c_b)+r_{c_b}$. Using Claim \ref{claim:rel_med_dists}, $\distkm(\mu_A,c_b)+r_{c_b}\leq d(\mu_A,b)+\distkm(b,c_b)+3r_{c_b}\leq d(\mu_A,b)+3\distkm(b,C_{i-1})$. Hence, $\distkm(b,C_{i-1})\geq (\distkm(\mu_A,C_{i-1})-d(\mu_A,b))/3=(v-\alpha r)/3$. We obtain the following.


\begin{align*}
      \tp(IR(A,\alpha),C_{i-1})&= \sum_{b \in IR(A,\alpha)}\distkm(b,C_{i-1})^2\\
      &\geq |IR(A,\alpha)|\frac{(v-\alpha r)^2}{9}\\
      &\geq |A|\left(1-\frac{1}{\alpha^2}\right)\frac{(v-\alpha r)^2}{9}
\end{align*}

Therefore, we have $\Pr(x \in IR(A,\alpha)\mid x \in A , A \in B_i)\geq \frac{(1-1/\alpha^2)(v-\alpha r)^2}{18(r^2+9 v^2)}$. Since $A$ is an unsettled optimal cluster with respect to $C_{i-1}$, we cannot have $v=\distkm(\mu_A,C_{i-1})$ to be very small. We obtain a lower bound on $v$ as follows. Since $A$ is an unsettled optimal cluster with respect to $C_{i-1}$, we have $\tp(A,C_{i-1})>20(1+\eps)\cdot OPT(A)$.
\begin{align}
        20(1+\eps) \cdot OPT(A) &<\tp(A,C_{i-1}) \label{eqn:lower-v-2} \\
         &\leq 2(OPT(A)+9 |A|v^2)) \notag
\end{align}
Simplifying, we get 
\begin{equation}\label{eqn:lower-simple-v}
     v>\sqrt{\frac{(18+20\eps)\cdot OPT(A)}{18|A|}}=\sqrt{(1+10\eps/9)} \cdot r
\end{equation}

Substituting the above value for $v$, we get the probability of sampling a point from the inner ring of an unsettled optimal cluster $A$, conditioned on sampling a point from $A$ is at least $\Pr(x \in IR(A,\alpha)\mid x \in A , A \in B_i)\geq (1-1/\alpha^2) \frac{(\sqrt{1+10\eps/9}-\alpha)^2}{180(1+\eps)}$. \end{proof}

For $\alpha =1+\eps/3$, we have $2(1+9\alpha^2)\leq 20(1+\eps)$. For this choice of $\alpha$, following Lemma \ref{lem:unsettle-settle}, we obtain that the cost of an optimal cluster becomes at most $20(1+\eps)$ times its optimal cost once $180\log n$ points are sampled from its inner ring. The probability of sampling a point from the inner ring of $A$ becomes $\Pr(x \in IR(A,\alpha)\mid x \in A , A \in B_i)\geq (1-\frac{1}{(1+\eps/3)^2}) \frac{(\sqrt{1+10\eps/9}-(1+\eps/3))^2}{180(1+\eps)}\geq \frac{\eps^3}{29160}=g(\eps)$, where $g(\eps)=\frac{\eps^3}{29160}$. We show that during iteration $i$ of Algorithm \ref{alg1}, a point sampled following the $\Tds$ distribution will belong to the inner ring of an unsettled optimal cluster with probability at least $g(\eps)/2$. We bound the number of required samples such that the sampled set $C$ contains at least $180 \log n$ points from the inner rings of all unsettled optimal clusters with high probability.

\begin{lem} \label{lem:sample-all-inner-ring} Let $C \subset X$ denote a set of $\frac{4320k\log n}{g(\eps)}$ points sampled using distribution $\Tds$ in Algorithm \ref{alg1}. Then, the total number of samples from the inner rings of unsettled optimal clusters in $C$ is at least $180k\log n$ with probability at least $1-\frac{1}{n^k}$. \end{lem}

\begin{proof} Let $C$ denote a set of $\frac{4320k\log n}{g(\eps)}$ points sampled following distribution $\Tds$ in Algorithm \ref{alg1}. We will show that $C$ contains at least $180k \log n$ points from the inner rings of unsettled optimal clusters with high probability. Let $s_i$ be a point sampled during iteration $i$. We define a random variable $X_i$ as follows:
 \[
      X_i = \begin{cases}
                        1 & \text{ if } s_i~\text{belongs to the inner ring of an unsettled optimal cluster} \\
                        0 & \text{otherwise}
                \end{cases}
\]

Let $l = \frac{4320k\log n}{g(\eps)}$. Then, $X = \sum_{i=1}^l X_i$ represents the total number of sampled points that belong to the inner rings of unsettled optimal clusters. From Lemma \ref{lem:progress} and Lemma \ref{lem:sample-inner-ring}, the probability that a point sampled using distribution $\Tds$ belongs to the inner ring of an unsettled optimal cluster is:
\begin{align*}
                \Pr[X_i=1] =& \Pr[x\in IR(A,\alpha): A\in B_i]\\
                =& \Pr[x \in IR(A,\alpha) \mid x \in A, A \in B_i] \cdot \Pr[x\in A, A\in B_i]\\
                \geq& \frac{1}{2} \cdot g(\eps)
\end{align*}

Therefore, we have $\EE[X_i] \geq \frac{1}{2} \cdot g(\eps)$. The expected number of centers in $C$ that belong to the inner rings of unsettled optimal clusters is $\EE[X] = \EE\left[\sum_{i=1}^l X_i\right] = \sum_{i=1}^l \EE[X_i] \geq l \cdot \frac{1}{2} \cdot g(\eps) =  2160k \log n$. Using Chernoff bounds (Lemma \ref{lem:chernoff}), the probability $\Pr[X <  180k \log n] = \Pr[X < \left(1 - \frac{11}{12}\right) 2160k \log n] \leq e^{-2160 k \log n \cdot \frac{121}{144} \cdot \frac{1}{2}} \leq \frac{1}{n^k}$.
\end{proof}

\begin{lem}\label{lem:num_itr} With probability at least $(1-1/n^k)(1-1/n^2)$, all optimal clusters become settled once a set $C$ of $\frac{4320k\log n}{g(\eps)}$ points are sampled following distribution $\Tds$ in Algorithm \ref{alg1}. \end{lem}

\begin{proof} Following Lemma \ref{lem:sample-all-inner-ring}, with probability at least $(1-1/n^k)$, a set $C$ of $\frac{4320k\log n}{g(\eps)}$ points sampled following distribution $\Tds$ contains at least $180k\log n$ points from the inner rings of unsettled optimal clusters. Using Lemma \ref{lem:unsettle-settle}, we know that once $180\log n$ points are sampled from the inner ring of an unsettled optimal cluster $A$, with probability at least $1-1/n^{3}$, that cluster becomes settled. We condition on the event that unsettled optimal clusters become settled once $180\log n$ points are sampled from their inner rings, and using a union bound over $k$ optimal clusters, this event happens for all $k$ optimal clusters with probability at least $(1-1/n^2)$. 

Note that as long as there is an unsettled optimal cluster, a sampled point in $C$ will belong to the inner ring of some unsettled optimal cluster with probability at least $g(\eps)/2$, and crucially, this probability is independent of the number of unsettled optimal clusters. Once an unsettled optimal cluster becomes settled with $180\log n$ points sampled from its inner ring, it stops contributing to the probability of sampling a point from the inner rings of unsettled optimal clusters. Subsequently, points sampled in $C$ with probability at least $g(\eps)/2$ would account for the number of points sampled from the inner rings of the remaining unsettled optimal clusters.

Since a sampled set $C$ of size $\frac{4320k\log n}{g(\eps)}$ contains $180k\log n$ points from the inner rings of unsettled optimal clusters with probability at least $(1-1/n^k)$, the probability that all optimal clusters become settled using sampled set $C$ is at least $(1-1/n^k)(1-1/n^2)$.  \end{proof}

\begin{lem}\label{lem:final-bound} Let $C\subset X$ be a set of $\frac{4320}{g(\eps)} \cdot k \log n$ sampled points. Then, we obtain $\tp(X,C) \leq 40(1+\eps) \cdot OPT$ with probability at least $(1-1/n^k)(1-1/n^2)(1-1/n)$. \end{lem}

\begin{proof} Using Lemma \ref{lem:sample-all-inner-ring}, we know that if we sample a center set $C$ of size $\frac{4320}{g(\eps)} k \log n$ following the $\Tds$-sampling distribution, the probability that $C$ contains at least $180 k \log n$ points from the inner rings of the optimal clusters is at least $1-\frac{1}{n^k}$. Following Lemma \ref{lem:num_itr}, for a sampled set $C$ of size at most $\frac{4320}{g(\eps)} \cdot k \log n$, all clusters become settled with probability at least $(1-1/n^k)(1-1/n^2)$. Using Lemma \ref{lem:dist_est_cor}, the probability that event $\E$ holds for all iterations of Algorithm \ref{alg1} is at least $(1-1/n)$. Hence, Algorithm \ref{alg1} gives a $40(1+\eps)$-approximation for $k$-means with probability at least $(1-1/n^k)(1-1/n^2)(1-1/n)$.

Query complexity: We initiate the algorithm with $180\log n$ centers, and at the end of the algorithm, we have in total $180\log n + \frac{4320\cdot 29160}{\eps^3} \cdot k \log n = O\left( \frac{k \log n}{\eps^3} \right)$ centers. In every iteration, we sample a center, and for that center, we query the weak oracle $\Wot$ for every $x \in X$ to get the distances. Hence, we make $n$ weak-oracle queries in each iteration. Since there are $O\left(\frac{k \log n}{\eps^3} \right)$ iterations, in total, we make $O\left(  \frac{nk \log n}{\eps^3} \right)$ queries to the weak-oracle. To compute $\distkm(x, C_i)$ for all $i$ in Algorithm \ref{alg1}, we need to perform strong-oracle queries. We make these strong-oracle queries to find the distances between any two centers in $C_i$. Hence, the number of strong-oracle queries required by Algorithm \ref{alg1} is at most $O(\frac{k^2\log^2 n}{\eps^6})$. 

With probability at least $(1-1/n^2)$, we sample at most $180 \log n$ points from the inner ring of any unsettled optimal cluster $A$. We require the number of points in the inner ring, $(1-1/\alpha^2)|A|$, to be at least $180 \log n$. This gives a lower bound on the size of each optimal cluster $A$ of $\frac{480\log n}{\eps}$, for $\alpha=1+\eps/3$.
\end{proof}

\subsection{Removing Lower Bound Assumption on Cluster Sizes}

In Section \ref{sec:k-means-weak-strong}, we showed that conditioned on each optimal cluster having at least $\frac{480\log n}{\eps}$ points, there exists an algorithm for $k$-means clustering that gives a $40(1+\eps)$-approximation with high probability. In this section, we describe how do adapt the algorithm when there are optimal clusters with fewer than $\frac{480\log n}{\eps}$ points. These are the optimal clusters $A$ for which the inner ring $IR(A,\alpha)$ contains fewer than $180\log n$ points in it. We know that following Lemma \ref{lem:sample-inner-ring} and Lemma \ref{lem:sample-all-inner-ring}, points from the inner ring of an unsettled optimal cluster are sampled with probability at least $g(\eps)/2$ and once we sample $\frac{4320}{g(\eps)} \cdot k \log n$ points, we sample $180\log n$ points from the inner rings of all optimal clusters. Now, if there are less than $180\log n$ points in the inner ring of some optimal clusters, it implies that we sample all points of those inner rings. However, since the number of points sampled from the inner ring is less than $180\log n$, we cannot apply Lemma \ref{lem:unsettle-settle} to claim that the unsettled optimal cluster $A$ would become settled.

We modify our algorithm as follows. Since there might exist small clusters with less than $\frac{480\log n}{\eps}$ points, we increase the total number of points to be sampled to ensure that $\frac{480\log n}{\eps}$ points will be sampled from any unsettled optimal cluster. The set $C$ of samples will have size $O(\frac{k\log n}{\eps g(\eps)})$. For optimal clusters with more than $180\log n$ points in their inner ring, as before, the Algorithm \ref{alg1} will sample $180\log n$ points from their inner rings and those clusters will become settled with high probability. For optimal clusters with fewer than $180\log n$ points in their inner rings, either all points of those clusters will be sampled or the cluster will become settled. Substituting $g(\eps)=O(\eps^3)$, we get the sample complexity of this modified algorithm to be $O(\frac{k\log n}{\eps^4})$, and for the same approximation guarantee of $40(1+\eps)$-approximation, the number of strong-oracle and weak-oracle queries will be $O(\frac{k^2\log^2 n}{\eps^8})$ and $O(\frac{nk\log n}{\eps^4})$, respectively.

\begin{thm}\label{thm:k-means-cons-biapprox} Let $\eps \in (0,1)$ and $\delta\leq 1/3$. There exists a randomized algorithm for $k$-means that makes $O(\frac{k^2 \log^2 n}{\eps^8})$ strong-oracle queries and $O(\frac{nk\log n}{\eps^4})$ weak-oracle queries to yield a $(O(\frac{\log n}{\eps^4},40(1+\eps)))$ bi-criteria approximation for $k$-means, and succeeds with at least some constant probability. \end{thm}
\begin{remark}\label{remark:km-better-constant}
   We note that in Algorithm \ref{alg1}, the constant i.e., $4320 \cdot 29160$ appearing in the value of $t$, before $\frac{k \log n}{\eps^3}$, can be improved by using the biased Cauchy–Schwarz inequality (Fact \ref{fact:bias-cauchy}) in place of the generalized triangle inequality in the analysis above. This change affects only the hidden constant in the big-$O$ notation and does not affect the statement of Theorem \ref{thm:k-means-cons-biapprox} for $k$-means. In the following, we indicate exactly where the biased Cauchy–Schwarz inequality should replace the generalized triangle inequality to achieve this improvement.

   Using the generalized triangle inequality in equation \ref{eqn:cost-a-ci}, the cost of an optimal cluster $A$ with respect to the current center set $C_i$ satisfies $\Tilde{\Phi}(A,C_i) \leq 2(1+9\alpha^2)\cdot OPT(A)$. On the other hand, applying the biased Cauchy-Schwarz inequality gives $\Tilde{\Phi}(A,C_i) \leq (\beta+1)\left(\frac{1}{\beta}+9\alpha^2\right)\cdot OPT(A)$. For $\beta = 1$, this recovers the same bound $\Tilde{\Phi}(A,C_i) \leq 2(1+9\alpha^2)\cdot OPT(A)$, but choosing $\beta < 1$ yields a strictly improved bound. Next, using the Cauchy--Schwarz inequality in Equation \ref{eqn:lower-v-1} and substituting the resulting bound on $\Tilde{\Phi}(A,C_i)$ into Equation \ref{eqn:lower-v-2} gives $v > \sqrt{\frac{20(1+\epsilon) - (\beta+1)\beta}{9(\beta+1)}}\, r$. When $\beta = 1$, this simplifies to $v > \sqrt{1 + \frac{10\epsilon}{9}}\, r$, which recovers  the earlier bound for $v$ in Equation \ref{eqn:lower-simple-v}. If we instead set $\beta = 0.15$, then the above expression becomes $v > \sqrt{1+\frac{20\eps}{10.35}}\, r$. We now choose $\alpha=\Big(1+\frac{2\eps}{3}\Big)$. Substituting the chosen values of $\beta$ and $\alpha$ into the new bound for $\Tilde{\Phi}(A,C_i)$ gives $\Tilde{\Phi}(A,C_i) \leq 20(1+\eps)\cdot OPT(A)$. Recall that when sampling a point according to the $\widetilde{D^2}$ distribution, the probability that a sampled point $x$ lies in the inner ring of an unsettled cluster $A$ satisfies $\Pr(x \in IR(A,\alpha)\mid x \in A , A \in B_i)\geq \frac{(1-1/\alpha^2)(v-\alpha r)^2}{18(r^2+9 v^2)}=\frac{(1-1/\alpha^2)(v/r-\alpha)^2}{18(1+9 v^2/r^2)}$. Substituting the chosen values of $\alpha$ and $v$, we obtain $\Pr(x \in IR(A,\alpha)\mid x \in A , A \in B_i)\geq \Big(1-\frac{1}{\left(1+\frac{2\eps}{3}\right)^2}\Big)\frac{\left(\sqrt{1+\frac{20\eps}{10.35}}-\left(1+\frac{2\eps}{3}\right)\right)^2}{180\left(1+18\eps/10.35\right)} > \frac{\eps^3}{13767}$. Since this probability is larger than what we had earlier, fewer centers need to be sampled, and consequently the required value of $t$ decreases to $t=4320 \cdot 13767$, where it was $t=4320 \cdot 29160$ earlier.

\end{remark}


\subsection{Constant Approximation for \texorpdfstring{$k$}{k}-means in Weak-strong Oracle Model}

\begin{thm} There exists a randomized algorithm for $k$-means in the weak-strong oracle model that computes a constant factor approximation with constant probability using $O(\frac{k^2 \log^2 n}{\eps^8})$ strong-oracle queries and $O(\frac{nk\log n}{\eps^4})$ weak-oracle queries. \end{thm}

\begin{proof} We use the algorithm of Theorem \ref{thm:k-means-cons-biapprox} to obtain a bi-criteria approximation guarantee for $k$-means. Let $C_{t+1}$ denote the resulting set of centers. Each center $c_i \in C_{t+1}$ is initialized with weight $w(c_i)=0$. Next, we assign every point $x$ to the center in $C_{t+1}$ for which $\distkm(x,C_{t+1})$ is minimized. This create a weighted $k$-means instance with $O(\frac{k \log n}{\eps^4})$ points where each center $c_i$ in $C_{t+1}$ is assigned $w(c_i)$ points. Since we have used strong-oracle queries for all pairs of points in this weighted instance, we know the exact distances between all pairs of points. We run a constant factor approximation algorithm, such as the algorithm by \cite{MP2004} on this weighted $k$-means instance to obtain a constant factor approximation for $k$-means. The constant factor approximation guarantee for $k$-means follows using Proposition $20$ of \cite{BDJW2024}. \end{proof}

\section{LOWER BOUND IN THE WEAK-ORACLE MODEL} \label{sec:weak-oracle}

In this section we show that any randomized algorithm giving a constant factor approximation for the $k$-means problem in the weak-oracle model requires to make expected $\Omega(\frac{nk}{(1-2\delta)^2})$ queries. For simplicity of calculations, we show this result for the $k$-median problem. Similar results follow for the $k$-means problem as well.

\begin{thm} Let $\delta \in (0,1/2)$. Any randomized algorithm in the weak-oracle model giving any constant factor approximation for the $k$-median problem with probability at least $3/4$ must make $\Omega(\frac{nk}{(1-2\delta)^2})$ of queries in expectation to the weak-oracle. \end{thm}

\begin{proof} To prove the lower bound, we construct a hard instance for which any randomized algorithm giving any constant factor approximation for $k$-median with high probability requires to make a lot of weak-oracle queries. Consider an instance $X$ with $n$ points. These $n$ points are partitioned into $k$ groups, each containing $\frac{n}{k}$ points. We denote these partitions as $C_1,C_2,\ldots,C_k$. We define distances between points in $X$ as follows: for any two points $x,y$ belonging to the same partition $C_i$, we set $d(x,y)=1$, and for points $x\in C_i$ and $y\in C_j$, where $i\neq j$, we set $d(x,y)=l$, where $l$ is a large number. These distances satisfy the metric property. For this instance, the optimal $k$-median solution has cost $n-k$.

Recall that an algorithm for $k$-median in the weak-oracle model gets access to distance values only through the weak-oracle. Suppose that all points in $X$ have been clustered correctly except for one point, $v$. We need to assign $v$ to its correct cluster, say $C_v$. If $v$ gets assigned to an incorrect cluster, then the cost for $v$ would be $l$. Hence, the $k$-median clustering cost would be at least $(n-k-1)+l$. For $l = c\cdot (n-k)$ for some constant $c>0$, the $k$-median cost of a clustering that misclassifies even a single point is at least some constant times the value of the optimal solution. Therefore, the algorithm needs to correctly classify all input points to achieve a constant factor approximation for $k$-median.

Let $\bar{d}(x,y)$ denote the value returned by the weak-oracle for any two input points $x,y$. Then, the following holds.
    \[
      \bar{d}(x,y)= 
        \begin{cases}
			1 & \text{with probability}~1-\delta \text{ if }~x,y~\text{belong to the same cluster}\\
			l & \text{with probability}~\delta   \text{ if }~x,y~\text{belong to the  same cluster}\\
			l & \text{with probability}~1-\delta  \text{ if }~x,y~\text{belong to different clusters}\\
		    1 & \text{with probability}~\delta    \text{ if }~x,y~\text{belong to different clusters}
        \end{cases}
    \]

For our query lower bound in the weak-oracle model, we use a lower bound result for the Query-Cluster problem studied by \cite{MS2017}. The authors consider a setup where there is an unknown ground truth $k$-clustering of a set $X$ of $n$ points, and the algorithmic task is to recover the ground truth clustering using queries to an oracle that answers whether any two points belong to the same cluster or not. Moreover, the query answers are wrong independently with probability $\delta<1/2$. The authors prove a $\Omega(\frac{nk}{(1-2\delta)^2})$ lower bound on the number of oracle queries for this Query-Cluster problem that recovers the ground truth $k$-clustering.

Our weak-oracle query lower bound of $\Omega(\frac{nk}{(1-2\delta)^2})$ follows by observing that one can reduce the Query-Cluster problem to the cluster recovery problem on an instance generated as above by the weak-oracle. Hence, the weak-oracle query lower bound $\Omega(\frac{nk}{(1-2\delta)^2})$ holds for any algorithm giving any constant factor approximation of $k$-median in weak-oracle model.\end{proof}

\section{ALGORITHM FOR $k$-CENTER IN WEAK-STRONG ORACLE MODEL} \label{sec:k-center-weak-strong}

We start by exploring whether known algorithms for the $k$-center problem can be implemented in the weak-strong oracle model. One key challenge is the {\em assignment} of points to centers. This means that even if we can locate $k$ good centers, we also need to output a mapping of points to the nearest center. This implies that we need to be able to calculate the distances of every point to a set of centers. Since we can use only a sublinear number of strong-oracle queries, we do not have accurate distances for doing this mapping. The lack of accurate distances also comes in the way of finding good centers since mapping and finding good centers go hand-in-hand in $k$-center algorithms. So, the important question is: {\em Can we get a reasonably accurate estimate of distances using weak-oracle queries?} We faced similar issues while designing an algorithm for the $k$-means problem as well. We will show that the distance function as defined in Definition \ref{def:km_dist_twopoints} will help in designing an algorithm for the $k$-center problem as well.



We will design a randomized algorithm within the weak-strong model that gives a $(6+\veps)$-approximation for the $k$-center problem. The main idea of the algorithm can be better understood in a simpler setting where the accurate distances are known. Let us see the outline of this algorithm. We shall assume that the optimal $k$-center radius $r_{opt}$ is known.\footnote{This is not an unreasonable assumption to make since we can guess the optimal radius with $(1\pm \veps)$ by iterating over discrete choices of the radius.} The algorithm is a simple ``{\em greedy ball-carving}'' procedure that is commonly used in the context of the $k$-center problem and can be stated as the following pseudo code (also stated as Algorithm 3 in \cite{BDJW2024}):

\begin{algorithm}[H]\label{greedy-ball} 
\caption{\tt Greedy Ball Carving} 
\SetKwInOut{Input}{Input}
\SetKwInOut{Output}{Output} 
\Input{Set of points $S$, radius $Rad$} 
\Output{A set of $C = \{c_1, \ldots, c_m\}$ centers and assignment of points in $S$ to centers in $C$}
$\triangleright$ Initialize: $C = \{\}$\\
\While{$S$ is not empty}{ 
            $\triangleright$ Pick an arbitrary point $c \in S$\\ 
            $\triangleright$ Treat $c$ as a center and assign any point $s \in S$ to $c$ if $d(c, s) \leq 2 Rad$.\\
            $\triangleright$ Add $c$ to $C$ and remove all points assigned to $c$ from $S$.         
}
\end{algorithm}

We execute the above algorithm with inputs $S = X$ and $Rad = r_{opt}$. We can argue that the greedy ball carving algorithm terminates after picking at most $k$ centers and assigning all points in $X$. This means that the algorithm gives a 2-factor approximation guarantee.

\begin{lem}\label{lem:ragesh-kc1} Let $(X, k)$ be any instance of the $k$-center problem and let $r_{opt}$ denote the optimal $k$-center radius for $X$. The greedy ball carving procedure, when executed with inputs $S = X; Rad = r_{opt}$, gives a $2$-approximate solution for the instance $(X, k)$. \end{lem}

\begin{proof} Let $\{c_1^*,\ldots, c_k^*\}$ be the optimal $k$ centers and let $X_1^*,\ldots, X_k^*$ be the corresponding $k$ clusters of dataset $X$. So, for every $i$, $X_i^* \subset \B(c_i^*, r_{opt})$. If the $k$ centers chosen by the greedy ball carving algorithm belong to $k$ distinct optimal clusters $X_1^*,\ldots,X_k^*$ and we assign points by carving out balls of radius $2r_{opt}$ around every chosen center, then all points are guaranteed to get assigned to a center within a distance $2r_{opt}$, giving a 2-factor approximation. So, all we need to show is that the algorithm chooses centers from all of the distinct optimal clusters. We can show this by induction on the iteration number. The statement trivially holds for the first iteration. Suppose the algorithm chooses centers from distinct clusters $j_1,\ldots,j_i$ in the first $i$ iterations. Note that in the next iterations, the algorithm picks centers from the set of points that does not contain any of the points in $X_{j_1}^*,\ldots, X_{j_i}^*$ as they have been carved out when removing points within balls of radius $2 r_{opt}$ around the chosen centers. So, the $(i+1)^{th}$ center will be chosen from a new optimal cluster. So, the ball carving algorithm terminates after picking $k$ centers from $k$ distinct optimal clusters. This completes the proof of the lemma.
\end{proof}

The reason we discussed the greedy ball carving algorithm and its analysis is that our $6(1 + \veps)$-factor approximation follows the same template. The only complication is that during the ball carving step, where we remove all points within the radius $2r_{opt}$, we need accurate distances, which we do not have in the weak-strong oracle model. So, we need to do more than choose and center, and carve out a ball of radius $2 r_{opt}$ in every step. Let us continue assuming that the optimal radius $r_{opt}$ is known. We will give a $6$-factor approximation under this assumption, which changes to $6(1 + \veps)$ when the assumption is dropped. To enable distance estimates, along with a center $c_i$, we must also pick a set of nearby points to $c_i$, which will help in estimating the distance of any point $y$ to $c_i$ following Lemma~\ref{lem:km_dist_twopoints}. A reasonable way to do this would be to uniformly sample a set $T_i$ of points (instead of one), use the strong-oracle queries on the subset and pick a center $c_i \in T_i$ with the largest number of points in $T_i$ within a radius $2 r_{opt}$. Let $S_{c_i}$ denote the subset of points in $T_i$ within a distance of $2r_{opt}$ to $c_i$. We can now obtain distance estimates using the tuple $(c_i, S_{c_i})$ using the weak-oracle queries as in Lemma~\ref{lem:km_dist_twopoints} and carve out those points and assign them to $c_i$ that satisfy $\distkm(c_i,z) \leq r$ for an appropriate value of $r$. Let $c_i$ belong to the optimal cluster $X_i^*$. What should the appropriate value of $r$ be such that the points in $S \cap X_i^*$ are guaranteed to get carved out and assigned to $c_i$? Since the distance estimates are inaccurate, we must use $r = 4 r_{opt}$. However, this may cause a point that is $6 r_{opt}$ away from $c_i$ to get carved out and assigned to $c_i$. This is the reason we obtain an approximation guarantee of $6$. This high-level analysis lacks two relevant details: (i) how do we eliminate the assumption that $r_{opt}$ is known, and (ii) probability analysis for all points being assigned a center within distance $6 r_{opt}$. We remove the assumption about $r_{opt}$ being known by iterating over discrete choices for the optimal radius $Rad = (1+\veps)^0, (1+\veps), (1+\veps)^2, \ldots$. We will argue that the ball carving succeeds for the choice $Rad = r_{opt} (1 + \veps)$ with high probability. So, the approximation guarantee we get is $6(1 + \veps)$. Instead of iterating over the possible choices of $Rad$ linearly, we can do that using binary search, which results in the running time and query complexity getting multiplied by a factor of $O\left(\log{\frac{\log{\Delta}}{\veps}}\right)$ due to this repeated ball carving (once for every choice of $Rad$), where $\Delta$ is the aspect ratio. Assuming $\Delta$ to be polynomially bounded (which is a popular assumption), the multiplicative factor becomes $O\left(\log{\frac{\log{n}}{\veps}}\right)$. In the remaining discussion, we state our algorithm more precisely and do the probabilistic analysis.

\RestyleAlgo{ruled}

\begin{algorithm}[H]
\caption{\tt Weak-Greedy Ball Carving}\label{weak-greedy-ball} 
\SetKwInOut{Input}{Input}
\SetKwInOut{Output}{Output} 
\Input{Set of points $S$, radius $Rad$} 
\Output{A set of $C = \{c_1, \ldots, c_m\}$ centers and assignment of points in $S$ to centers in $C$}
$\triangleright$ Initialize: $C = \{\}$\\
\While{$S$ is not empty}{
	   \ \ (1) If ($|C| = k$) {\bf abort}\\
	   (2) If ($|S| \leq 180 k \log{n}$), use strong-oracle queries to find the remaining centers covering all points in $S$ within distance $2 Rad$. If this is not possible, {\bf abort}. Otherwise, output the $k$ centers picked.\\
            (3) Pick a subset $T \subset S$ of size $180 k \log{(n)}$ uniformly at random\\ 
            (4) Query the strong-oracle to find distances between elements in $T$ and use these to find a center $c \in T$ such that $|S_c|$ is maximized, where $S_c \equiv |\B(c, 2Rad) \cap T|$.\\
            (5) If ($|S_c| < 180 \log{n}$) {\bf abort}\\
            (6) Retain $180 \log{n}$ points in $S_c$ ({\it i.e., remove the extra point})\\
            (7) Treat $c$ as a center and assign any point $s \in S$ to $c$ if $\distkm(s,c) \leq 4 Rad$.\\
            (8) Add $c$ to $C$ and remove all points assigned to $c$ from $S$.      
}
\Return{$C$ and the assignment}
\end{algorithm}

We prove the following lemma related to Algorithm \ref{weak-greedy-ball}.

\begin{lem}\label{lem:ragesh-kc3} Let $(X, k)$ be any instance of the $k$-center problem and let $r_{opt}$ denote the optimal $k$-center radius for $X$. Then, Algorithm \ref{weak-greedy-ball} satisfies the following two properties: 
\begin{enumerate}
    \item For any value of $Rad$, given that the algorithm does not abort, the probability that all points are within $6 Rad$ of the centers chosen is at least $(1-1/n^4)$.
    \item  For any $Rad \geq (1 + \veps) r_{opt}$, the probability that the algorithm does not abort is at least $(1-1/n^4)$.
\end{enumerate}
\end{lem}

\begin{proof} For the second property, we will use the same induction-based argument as in the proof of Lemma~\ref{lem:ragesh-kc1}.  Given that in the first $i-1$ iterations, the centers $c_1,\ldots,c_{i-1}$ belong to $i-1$ distinct optimal clusters $X_{j_1}^*,\ldots, X_{j_{i-1}}^*$ and $S$ does not contain any points from these clusters, we will show that the probability that the next chosen center $c_{i}$ belonging to a new cluster $X_{j_{i}}^*$ carves out any point  $z \in S \cap X_{j_{i}}^*$ is at least $(1 - 1/n^5)$. Note that by our choice of center $c_i$ and the fact that $Rad \geq (1+\veps)r_{opt}$, we have $|S_{c_i}| \geq 180 \log{n}$. 
Consider any point $z \in X_{j_i}^* \cap S$. We have $d(z, c_i) \leq 2 Rad$ and from Lemma~\ref{lem:km_dist_twopoints}, $\Pr[|\distkm(z, c_i) - d(z, c_i)| \leq 2 Rad] \geq 1 - \frac{1}{n^5}$. From this we get $\Pr[\distkm(z,c_i)\leq 4 Rad] \geq 1 - \frac{1}{n^5}$. So, with probability at least $(1-\frac{1}{n^5})$, $z$ gets assigned to the center $c_i$. So, the probability that there is a point in $X_{j_i}^* \cap S$ that does not get assigned to $c_i$ is at most $\frac{1}{n^5}$. So, by union bound, the probability of a point in $X$ that does not get assigned to centers $c_1,\ldots,c_k$ is at most $\frac{1}{n^4}$.

For the first property, let us bound the probability that a point $z$ gets assigned a center that is at a distance $> 6 Rad$. Suppose the $z$ is assigned a center in iteration $i$. So, it gets assigned center $c_i$. This means that $\distkm(z, c_i) \leq 4 Rad$. The probability that $d(z, c_i) > 6 Rad$ is at most $\frac{1}{n^4}$ from the previous discussion. So, the probability that there is a point that gets assigned a center which is $> 6 Rad$ distance away is at most $\frac{1}{n^4}$ from the union bound. So, property (1) follows.
\end{proof}

Let us see why calling Algorithm \ref{weak-greedy-ball} with input $(X, Rad)$ while doing a binary search over $Rad$ gives a $6(1+\veps)$-approximate solution with high probability. 
From Lemma~\ref{lem:ragesh-kc3}, if $Rad < (1+\veps) r_{opt}$, and the algorithm does not abort, then with probability at least $(1-1/n^4)$ all points are assigned a center that is within distance $6 Rad < 6(1+\veps) r_{opt}$. On the other hand, if $Rad \geq (1+\veps)r_{opt}$, then the algorithm does not abort with probability at least $(1-1/n^4)$. So, the binary search will terminate with a $6(1+\veps)$-approximate solution with high probability.

\underline{\em Query complexity}: The number of strong-oracle queries in every iteration of the {\tt weak greedy ball carving} algorithm is $O(k^2\log^2{n})$. So, the total number of strong queries across $k$ iterations is $O(k^3 \log^2{n})$. The number of weak Oracle queries is $O(nk \log{n})$. Taking into account the binary search over $Rad$, the number of strong-oracle and weak-oracle queries are $O \left(k^3 \log^2{n} \log{\frac{\log{n}}{\veps}} \right)$ and $O\left(nk\log{n} \log{\frac{\log{n}}{\veps}} \right)$, respectively.

\begin{thm} There exists a $6(1+\eps)$-approximation algorithm for the $k$-center problem that makes $O \left(k^3 \log^2{n} \log{\frac{\log{n}}{\veps}} \right)$ strong-oracle and $O\left(nk\log{n} \log{\frac{\log{n}}{\veps}} \right)$ weak-oracle queries and succeeds with probability at least $(1-1/n^4)^2$. \end{thm}
\begin{remark}\label{remark:kc-better-so}
          We highlight that for $k\ll \log n$, one can obtain better strong-oracle query bounds for $k$-center clustering. We sketch the main idea below. We sample a set $T$ of $O(k\log n)$ points uniformly at random. Next, we subsample $T$ to obtain $Y$, where each $x\in T$ is included in $Y$ independently with probability $1/k$. Using linearity of expectation, we have $\mathbb{E}[|Y|]=O(\log n)$. Consider any large optimal cluster $C^*$ such that sample $T$ contains at least $O(\log n)$ points of that cluster. Since $k\ll \log n$, with high probability, at least one point from the optimal cluster $C^*$ will be sampled in $Y$. We think of the set $Y$ as the set of {\em good} centers. Since $\mathbb{E}[|Y|]=O(\log n)$ and for each $y\in Y$, we make strong-oracle queries to all $x\in T$, overall this makes $O(k\log^2 n)$ strong-oracle queries in expectation. The algorithm repeats the above step at most $k$ times. We repeat the above steps at most $O(\log \frac{\log n}{\epsilon})$ times for the guess of the optimal radius $Rad$. Hence, overall the algorithm makes at most $O(k^2\log^2 n\log \frac{\log n}{\epsilon})$ strong-oracle queries in expectation. Using standard concentration bounds, one can show that the algorithm makes $O(k^2\log^2 n\log \frac{\log n}{\epsilon})$ strong-oracle queries with high probability as well.
    \end{remark}

\section{EXPERIMENTAL RESULTS} \label{sec:detailed-experiments}

We run our algorithms for $k$-means and $k$-center problems on synthetic as well as real-world datasets to demonstrate that our algorithms can be implemented efficiently in practice, and in particular, our algorithms outperform the algorithms of \cite{BDJW2024}. In our experiments, we vary the failure probability $\delta$ of the weak-oracle and run our algorithms on datasets of different sizes, and report how the cost of clustering varies with the number of oracle queries. 

\begin{remark}\label{remark:strong-point-rel} In our discussions of strong-oracle queries so far, we interpret a strong-oracle query as a query for the distance between two points. We note that \cite{BDJW2024}, in their experiment results, report the number of {\em point} queries, where an oracle on a point query provides the embedding of the point. Consider the following example for more clarity about the relation between the number of point queries and the number of strong-oracle queries. Let us have a set $S$ of points, and let us make a point query for each of the $|S|$ points. This will give us the embedding of all points in the set $S$, and using these embeddings, we can compute the pairwise distances between all pairs of points in $S$. On the other hand, to find the pairwise distances between all pairs of points in $S$, one must make $|S|\choose 2$ strong-oracle queries.    

We note that \cite{BDJW2024} reports their experimental results using the number of point queries. We convert the reported point query counts to strong-oracle query counts using the relationship mentioned above, and compare the number of strong-oracle queries used in \cite{BDJW2024} with those obtained in our experiments.

\end{remark}

\paragraph{Experimental setup} The initial experiments were carried out on a Mac Mini with an M2 chip and 8GB  RAM. Subsequently, the computationally intensive experiments were performed on a server having 1.5 TB  RAM with 64 CPU cores. It took about $10$ days to obtain the experimental results.

\paragraph{Datasets} We run our algorithms on synthetic as well as real-world datasets, and compare the performance of our algorithms with those of \cite{BDJW2024}. The datasets used in our experiments are described as follows. We generate synthetic data using the Stochastic Block Model (SBM) \citep{MNS2015, A2018}. We use three synthetic datasets in our experiments with the number of points $n$ being $10k, 20k$ and $50$k, respectively. Each of these datasets has $k=7$ clusters, and the points in the datasets belong to a seven-dimensional space. The points in the $i$th cluster of the datasets are drawn from a Gaussian distribution $N(\mu^{(i)},I)$, where $\mu^{(i)}[i]=10^5$ and $\mu^{(i)}[j]=0$ for $j\neq i$. This ensures that the points within the same cluster are close to each other, while the optimal clusters remain well separated.

For the real-world data, we use the MNIST dataset \citep{D2012,MH2008}. The MNIST dataset has ten clusters, $k=10$. The dataset has $n=60k$ images, each with a dimension of $28 \times 28$. Similar to \cite{BDJW2024}, we apply SVD and t-SNE embeddings to embed the dataset into $d=50$ dimensions and $d=2$ dimensions, respectively. 


\paragraph{Construction of the weak-oracle} In order to construct the weak-oracle, we first create a perturbed distance matrix for both datasets as follows. This matrix is initialized with the actual distances between the points in the dataset. For the SBM-based synthetic dataset, for any two points $i$ and $j$, if they belong to the same cluster, independently with probability $\delta$, the $(i,j)$th entry of the perturbed matrix is set to $10^5$. And if $i$ and $j$ belong to different clusters, independently with probability $\delta$, the $(i,j)$th entry of the matrix is set to one. This is how we create the perturbed matrix using which the weak-oracle responds to a query for the SBM-based synthetic dataset. 

For the MNIST dataset, if two points belong to the same cluster (i.e., the same digit class), we replace the distance between the points in the matrix with a randomly chosen inter-cluster distance between two points in different clusters. If the points belong to different clusters (i.e., different digit classes), we assign them a randomly chosen intra-cluster distance. \cite{BDJW2024} notes that the SBM and MNIST datasets embedded using t-SNE with $d = 2$ and SVD with $d = 50$, have clear ground-truth clusterings. Since we know the ground-truth clustering for these embeddings, we can construct the weak-oracle in the manner described above.

\paragraph{Baselines in our experiments} We run the $k$-means++ algorithm on these datasets using the strong-oracle, and use its output as the {\em strong baseline} in our experiments for $k$-means. Similarly, we run the $k$-means++ algorithm on these datasets using the weak oracle, and use its output as the {\em weak baseline}. In our experimental results, we compute the approximation factor of our algorithm by calculating the ratio of the cost of our algorithm with this strong baseline. 

We compare the performance of our algorithms for $k$-center with the baselines computed as follows. The output of the farthest point algorithm of  \cite{G1986} gives us the baseline for both the SBM-based synthetic dataset as well as the MNIST dataset. The strong and weak baselines for the experiments on $k$-center are obtained by running the algorithm of \cite{G1986} with the strong-oracle and weak-oracle, respectively. In the experiments, we report the approximation factors obtained by our algorithms as the ratio of the cost of our algorithm with the strong baseline. 

\paragraph{Experimental results on $k$-means clustering}
In our experiments, we run our $k$-means algorithm on SBM-based synthetic datasets of size $n=10k, 20k$ and $50k$, and set the failure probability $\delta$ of the weak-oracle to be $0.1,0.2$ and $0.3$. For each of the above choices, we run our algorithm with the number of strong-oracle queries, decided as follows. Using our theoretical results, $O\left(\frac{k^2\log^2 n}{(1/2-\delta)^2}\right)$ strong-oracle queries suffice to get a constant approximation. In the experiments, we vary the constant factor in the expression $ O\left(\frac{k^2 \log^2 n}{(1/2 - \delta)^2}\right) $ from $0.0001$ to $10$ and observe the resulting clustering cost.

\begin{table}[htbp]
   \centering
    \scriptsize
     \caption{Performance of our $k$-means algorithm on MNIST dataset }
    \begin{tabular}{|c|p{1cm}|p{1cm}|p{1cm}|p{1cm}|p{1cm}|p{1cm}|p{1cm}|p{1cm}|p{1cm}|p{1cm}|}
        \hline
        \textbf{MNIST} & \multicolumn{3}{c|}{\makecell{$\%$ of strong-oracle queries \\
            (Ours)}} & \multicolumn{3}{c|}{Approximation factor} & \multicolumn{3}{c|}{\makecell{$\%$ of strong-oracle queries \\
            \citep{BDJW2024}}} & $\%$ dropped \\ \hline
        & $\delta=0.1$ & $\delta=0.2$ & $\delta=0.3$ 
        & $\delta=0.1$ & $\delta=0.2$ & $\delta=0.3$
        & $\delta=0.1$ & $\delta=0.2$ & $\delta=0.3$
        & Value \\ \hline
        
        SVD   & $0.00013$ & $0.00019$ & $0.00024$ 
              & $1.019$ & $1.03$ & $1.027$ 
              & $0.00062$ & $0.00096$ & $0.00063$ 
              & $61$ \\\hline
        
        t-SNE & $0.00054$ & $0.00079$ & $0.00089$ 
              & $1.1595$ & $1.1433$ & $1.0960$ 
              & $0.20969$ & $0.20877$ & $0.43814$ 
              & $99$ \\ \hline
    \end{tabular}
   
    \label{tab:mnist-k-means}

\end{table}

\begin{table}[htbp]
    \scriptsize
    \centering
     \caption{Performance of our $k$-means algorithm on SBM-based synthetic data}
    \begin{tabular}{|c|p{1cm}|p{1cm}|p{1cm}|p{1cm}|p{1cm}|p{1cm}|p{1cm}|p{1cm}|p{1cm}|p{1cm}|}
        \hline
        $n$ & \multicolumn{3}{c|}{\makecell{$\%$ of strong-oracle queries \\
            (Ours)}} & \multicolumn{3}{c|}{Approximation factor} & \multicolumn{3}{c|}{\makecell{$\%$ of strong-oracle queries \\
            \citep{BDJW2024}}} & $\% $ dropped \\ \hline
        & $\delta=0.1$ & $\delta=0.2$ & $\delta=0.3$ 
        & $\delta=0.1$ & $\delta=0.2$ & $\delta=0.3$
        & $\delta=0.1$ & $\delta=0.2$ & $\delta=0.3$
        & Value \\ \hline
        10k   & $0.00585$ & $0.00616$ & $0.00616$  & $2.156$   & $0.9664$  & $0.998$ &  $0.3075$&  $0.12286$& $1.73861$  &$94$  \\ \hline
        20k   & $0.00200$ & $0.00218$ &$0.00242$  & $1.065$  & $1.089$ & $1.0695$ & $0.03890$ & $0.03159$ & $0.73405$ & $93$  \\ \hline
        50k   & $0.00039$ & $0.00043$ &$0.00045$  &$1.043$  & $1.058$  & $1.015$ & $0.01075$ &$0.00670$ & $0.05480$ & $93$ \\ \hline
    \end{tabular}
   
    \label{tab:sbm-k-means}

\end{table}

Tables \ref{tab:mnist-k-means} and \ref{tab:sbm-k-means} show our experimental results for $k$-means clustering on the MNIST dataset with SVD and t-SNE embeddings, and on SBM-based datasets of sizes $n = 10k$, $20k$, and $50k$, and present a comparison of the percentage of strong-oracle queries used by our algorithms and that of \cite{BDJW2024}.  We conducted experiments for different combinations of failure probability $\delta$ and the number of strong-oracle queries. In Table \ref{tab:sbm-k-means}, we mention the results for the setting for which the value of $(\text{number of strong-oracle queries} \times \log (\text{cost of clustering}))$ is minimized.

\noindent \textbf{Comparison of performance of our $k$-means algorithm with \cite{BDJW2024}}: We compare the performance of our $k$-means algorithm with the $k$-means algorithm of \cite{BDJW2024} as follows. Figures \ref{fig:km_sbm_10k}, \ref{fig:km_sbm_20k} and \ref{fig:km_sbm_50k} capture how the clustering cost varies with the number of strong-oracle queries when these algorithms are run on SBM-based synthetic datasets of size $10k$, $20k$ and $50k$, respectively. The blue, green, and red lines in the plots represent results for $\delta = 0.1$, $\delta = 0.2$, and $\delta = 0.3$, respectively. The solid lines show how the clustering cost behaves as the number of strong-oracle queries is changed in our algorithm. The dotted lines represent the results from \cite{BDJW2024}. In Figures \ref{fig:km_sbm_10k}, \ref{fig:km_sbm_20k}, and \ref{fig:km_sbm_50k}, the solid blue and green lines are not visible because they are hidden behind the red lines. This happens because, for all values of $\delta$, our percentage of strong oracle queries stays around $0.005$, which is much lower compared to the number of strong oracle queries used in \cite{BDJW2024}. As a result, all the blue and green lines appear overlapped by the red lines. Note that since we don't have their exact data, we reported these values by visually interpreting their plots.

We observe in the experiments that as the number of strong-oracle queries increases, the clustering cost given by our algorithm as well as the algorithm in \cite{BDJW2024} decreases. Take any value of $\delta$, say $\delta=0.1$. The cost versus the number of strong-oracle queries behavior of our algorithm is represented by the solid blue line, while the performance of the algorithm of \cite{BDJW2024} is shown using the dotted blue line in Figures \ref{fig:km_sbm_10k}, \ref{fig:km_sbm_20k} and \ref{fig:km_sbm_50k}. We observe that in both Figures \ref{fig:km_sbm_10k}, \ref{fig:km_sbm_20k} and \ref{fig:km_sbm_50k} our algorithm achieves a lower cost with significantly fewer strong-oracle queries. Similar observations hold for other values of $\delta$. Our algorithm makes at least $93\%$ fewer strong-oracle queries compared to \cite{BDJW2024} for every choice of $\delta \in \{0.1, 0.2, 0.3\}$ and $n \in \{10k, 20k, 50k\}$.

We note that the costs of $k$-means solutions of \cite{BDJW2024}, given enough strong-oracle queries, as shown in Figures \ref{fig:km_sbm_10k} and \ref{fig:km_sbm_20k}, tend to be smaller than the cost achieved by our $k$-means algorithm. It seems that the strong baseline used by \cite{BDJW2024} achieves a lower cost compared to the strong baseline used by our algorithm. Note that we use $
k$-means++ with strong-oracle queries to compute the strong baseline for $k$-means.

\begin{figure}[ht]
    \centering

    \begin{minipage}{0.33\linewidth}
        \centering
        \captionsetup{font=small, justification=centering}
        \caption{$k$-means on SBM dataset with $n=10k$}
        \includegraphics[width=\linewidth]{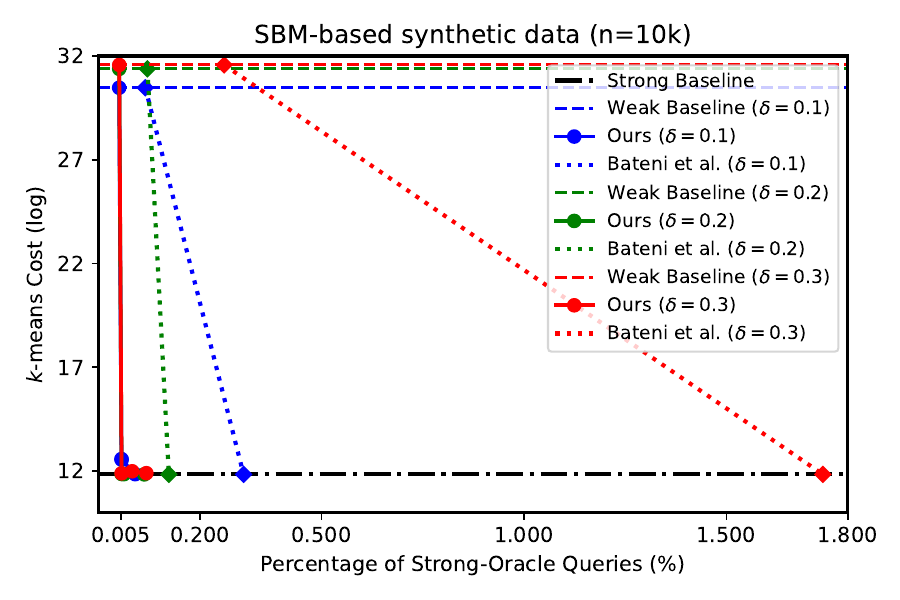}
        \label{fig:km_sbm_10k}
    \end{minipage}%
    \hfill
    \begin{minipage}{0.33\linewidth}
        \centering
        \captionsetup{font=small, justification=centering}
        \caption{$k$-means on SBM dataset with $n=20k$}
        \includegraphics[width=\linewidth]{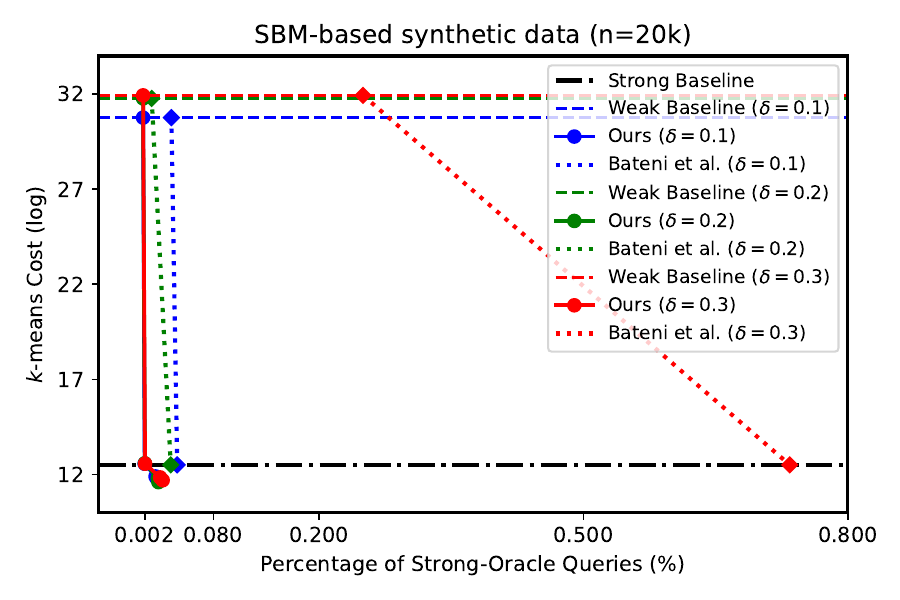}
        \label{fig:km_sbm_20k}
    \end{minipage}%
    \hfill
    \begin{minipage}{0.33\linewidth}
        \centering
        \captionsetup{font=small, justification=centering}
        \caption{$k$-means on SBM dataset with $n=50k$}
        \includegraphics[width=\linewidth]{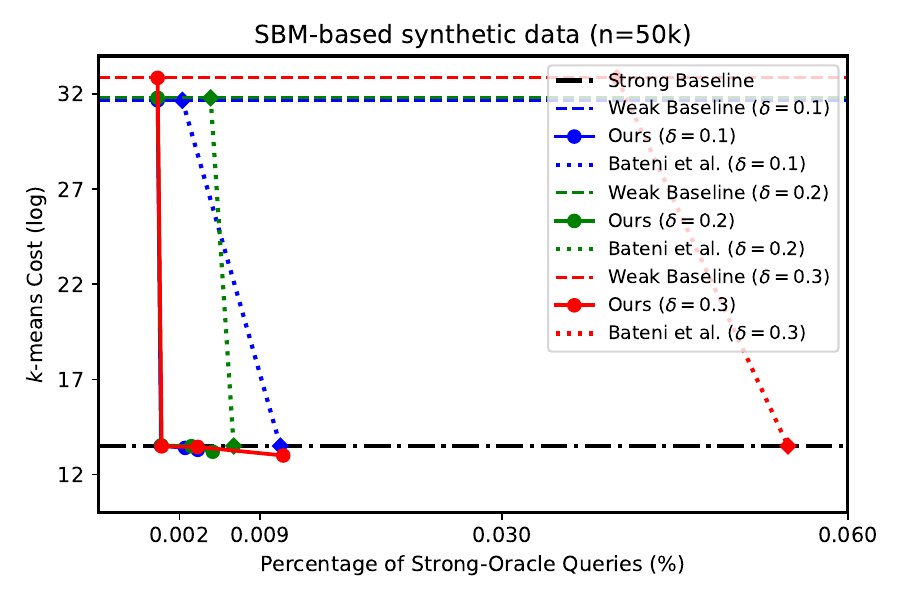}
        \label{fig:km_sbm_50k}
    \end{minipage}

\end{figure}

Figures~\ref{fig:km_mnist_svd} and~\ref{fig:km_mnist_tsne} show how our $k$-means clustering cost changes with respect to the number of strong-oracle queries on the MNIST dataset using SVD and t-SNE embeddings. Unlike SBM-based datasets, the distances between clusters in the MNIST are relatively smaller. Because of this, there is little difference between the weak and strong baselines in these figures. Recall that to construct the weak oracle, if two points belong to the same cluster, we replace their distance with the distance between two points from different clusters. Since the inter-cluster distances in MNIST (with both SVD and t-SNE embeddings) are small, the weak and strong baselines end up being quite similar.

We believe that the $k$-means++ cost on the MNIST dataset using t-SNE embeddings is higher than the cost using SVD embeddings. However, based on the visualizations of $k$-means results on MNIST provided in \cite{BDJW2024}, the baseline cost for t-SNE appears to be lower than that for SVD. Therefore, we could not directly use the results from \cite{BDJW2024} for comparison with our own.
 \begin{figure}[ht]
    \centering

    \begin{minipage}{0.5\linewidth}
        \centering
        \captionsetup{font=small, justification=centering}
        \caption{ $k$-means on MNIST with SVD embedding\\[0.3em]$n = 60k$ }
        \includegraphics[width=\linewidth]{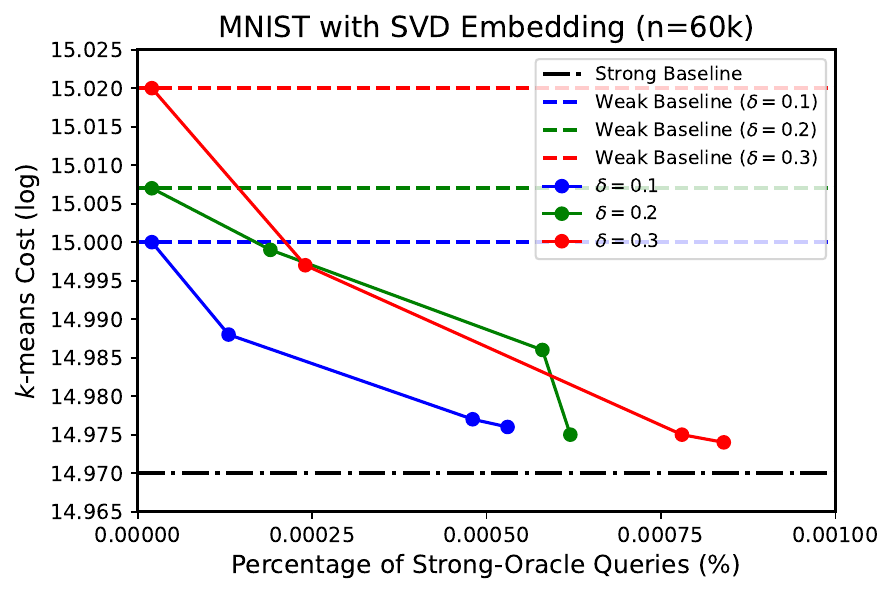}
        \label{fig:km_mnist_svd}
    \end{minipage}%
    \hfill
    \begin{minipage}{0.5\linewidth}
        \centering
        \captionsetup{font=small, justification=centering}
        \caption{$k$-means on MNIST with t-SNE embedding\\[0.3em]$n = 60k$}
        \includegraphics[width=\linewidth]{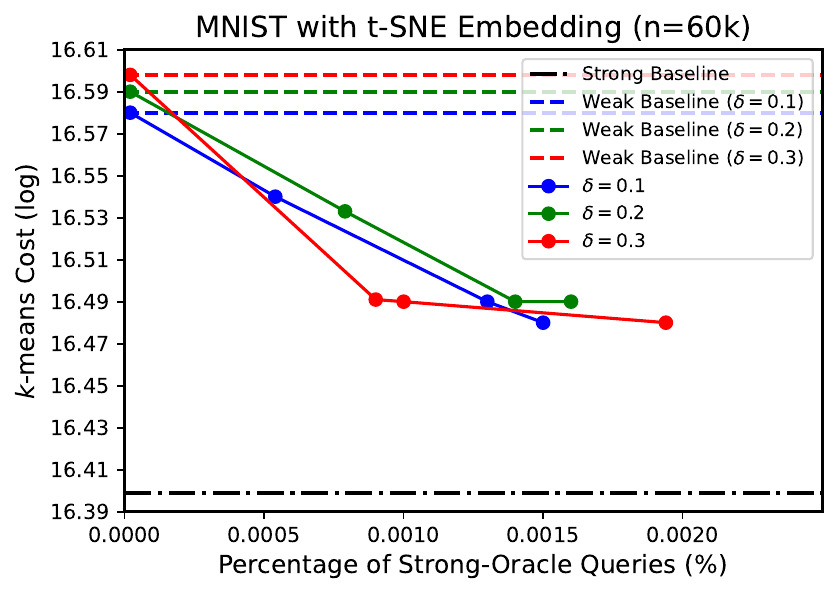}
        \label{fig:km_mnist_tsne}
    \end{minipage}

\end{figure}

\paragraph{Experimental results on $k$-center clustering}
We run our $k$-center algorithm on SBM-based synthetic datasets of size $n=10k, 20k$ and $50k$ as well as the MNIST dataset with SVD and $t$-SNE embeddings, and vary the failure probability $\delta$ of the weak-oracle as $0.1,0.2$ and $0.3$. For each of the above choices, we run our $k$-center algorithm on the dataset with varied number of strong-oracle queries. The number of strong-oracle queries 
used in the experiments is decided as follows. Using our theoretical results, $\tilde{O}(\frac{k^3\log^2 n}{(1/2-\delta)^2})$ strong-oracle queries suffice to get a constant approximation. In the experiments, we set the number of strong-oracle queries to be $O(\frac{k^3\log^2 n}{(1/2-\delta)^2})$ and vary the constant factor in the above expression from $0.0001$ to $10$ and observe the cost of clustering mentioned below.

\begin{table}[htbp]
    \scriptsize
    \centering
    \caption{Performance of our $k$-center algorithm on SBM-based synthetic data}
    \begin{tabular}{|c|p{1cm}|p{1cm}|p{1cm}|p{1cm}|p{1cm}|p{1cm}|p{1cm}|p{1cm}|p{1cm}|p{1cm}|}
        \hline
        $n$ & \multicolumn{3}{c|}{\makecell{$\%$ of strong-oracle queries \\
            (Ours)}} & \multicolumn{3}{c|}{Approximation factor} & \multicolumn{3}{c|}{\makecell{$\%$ of strong-oracle queries \\
            \citep{BDJW2024}}} & $\% $ dropped \\ \hline
        & $\delta=0.1$ & $\delta=0.2$ & $\delta=0.3$ 
        & $\delta=0.1$ & $\delta=0.2$ & $\delta=0.3$
        & $\delta=0.1$ & $\delta=0.2$ & $\delta=0.3$
        & Value \\ \hline
        10k   & $0.03980$   & $0.06225$  & $0.148625$ & $0.84$  & $0.9$   & $0.86$  &$0.42319$   & $0.54987$  & $0.54987$& $72$ \\ \hline
        20k   & $0.01204$& $0.019531$  & $0.04399$  & $0.78$  & $0.84$  & $0.79$  & $0.01581$ & $0.03832$ & $0.29785$ & $23$ \\ \hline
        50k   & $0.00192$ & $0.00358$ & $0.00808$  & $0.81$  & $0.8$   & $0.77$  & $0.00635$ & $0.02199$ & $0.04773$ & $83$ \\ \hline
    \end{tabular}
    
    \label{tab:sbm-k-center}

\end{table}

\begin{table}[htbp]
\begin{normalsize}
    \centering
    \caption{Performance of our $k$-center algorithm on MNIST dataset}
    \begin{tabular}{|c|p{1.2cm}|p{1.2cm}|p{1.2cm}|p{1.2cm}|p{1.2cm}|p{1.2cm}|}
        \hline
        MNIST & \multicolumn{3}{c|}{$\% $ of strong-oracle queries} & \multicolumn{3}{c|}{Approximation factor} \\ \hline
        & $\delta=0.1$ & $\delta=0.2$ & $\delta=0.3$ & $\delta=0.1$ & $\delta=0.2$ & $\delta=0.3$ \\ \hline
        t-SNE  & $0.00062$   & $0.00083$  & $0.00089$ & $2.46$  & $2.45$ & $2.31$ \\ \hline
        SVD & $0.00062$   & $0.00084$   & $0.00094$ & $1.276$  & $1.3$ & $1.37$  \\ \hline 
        
    \end{tabular}
    
    \label{tab:mnist-k-center}
\end{normalsize}
\end{table}

Tables \ref{tab:sbm-k-center} and \ref{tab:mnist-k-center} show our experimental results on $k$-center clustering on SBM-based synthetic datasets of size $n=10k, 20k$ and $50k$ and on MNIST dataset with SVD and $t$-SNE embeddings and present  comparison of the percentage of strong-oracle queries used by our algorithms and that of \cite{BDJW2024}. We conducted experiments for different combinations of $\delta$ and the number of strong-oracle queries. In Table \ref{tab:sbm-k-center} and Table \ref{tab:mnist-k-center}, we report the results for which the value of $(\text{number of strong-oracle queries} \times \log (\text{cost of clustering}))$ is minimized.

\noindent \textbf{Comparison of performance of our $k$-center algorithm with \cite{BDJW2024}}: We compare the performance of our $k$ center algorithm with the $k$-center algorithm of \cite{BDJW2024}. Figures \ref{fig:kc_sbm_10k}, \ref{fig:kc_sbm_20k}, and \ref{fig:kc_sbm_50k} illustrate how the clustering cost changes as the number of strong-oracle queries increases. These results are based on SBM-based synthetic datasets of size $10k$, $20k$, and $50k$, respectively. In the plots, the blue, green, and red lines represent the results for $\delta = 0.1$, $\delta = 0.2$, and $\delta = 0.3$. Solid lines show the performance of our algorithm, while dotted lines indicate the results from \cite{BDJW2024}. Since we do not have access to their exact data, we estimated their values by visually interpreting their plots. We observe significant reduction in the number of strong-oracle queries by our algorithm compared to \cite{BDJW2024}. Our results show that our algorithms use at least $72\%, 23\%$ and $83\%$ fewer strong-oracle
 compared to \cite{BDJW2024} for $n=10k, 20k$ and $50k$, respectively.

\begin{figure}[ht]
    \centering
    \begin{minipage}{0.33\linewidth}
        \captionsetup{font=small, justification=centering}
        \caption{$k$-center on SBM dataset with $n=10k$}
        \centering
        \includegraphics[width=\linewidth]{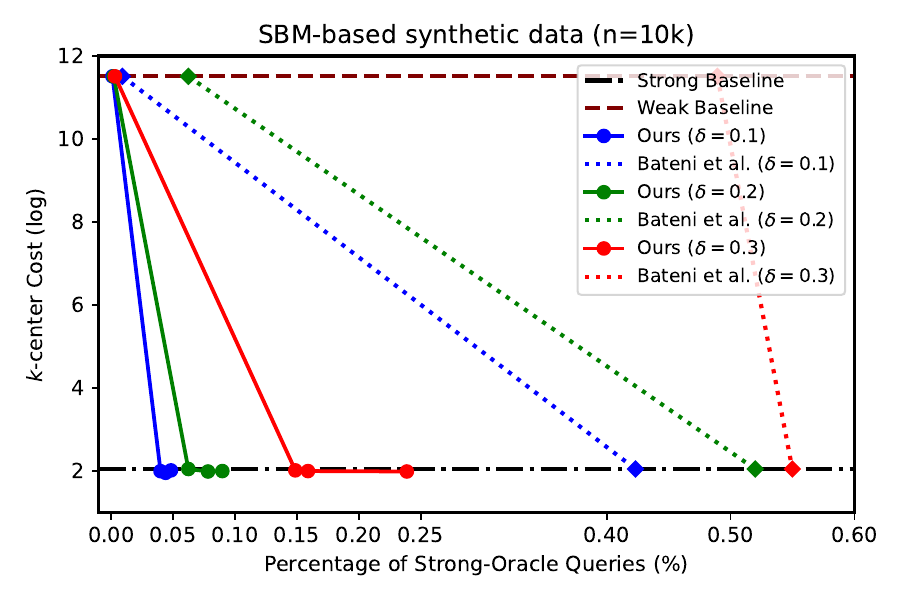}
        \label{fig:kc_sbm_10k}
    \end{minipage}%
    \hfill
    \begin{minipage}{0.33\linewidth}
        \captionsetup{font=small, justification=centering}
        \caption{$k$-center on SBM dataset with $n=20k$}
        \centering
        \includegraphics[width=\linewidth]{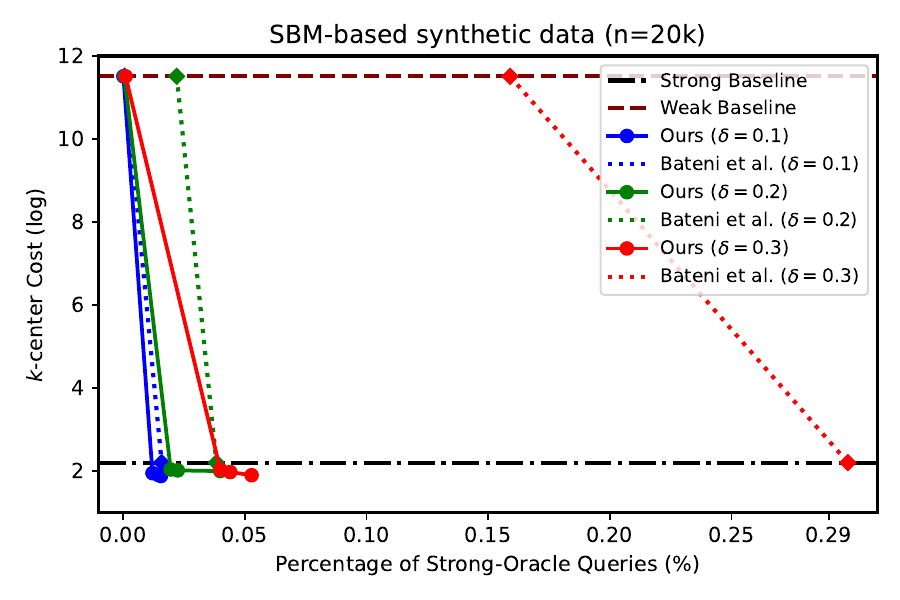}
        \label{fig:kc_sbm_20k}
    \end{minipage}%
    \hfill
    \begin{minipage}{0.33\linewidth}
        \captionsetup{font=small, justification=centering}
        \caption{$k$-center on SBM dataset with $n=50k$}
        \centering
        \includegraphics[width=\linewidth]{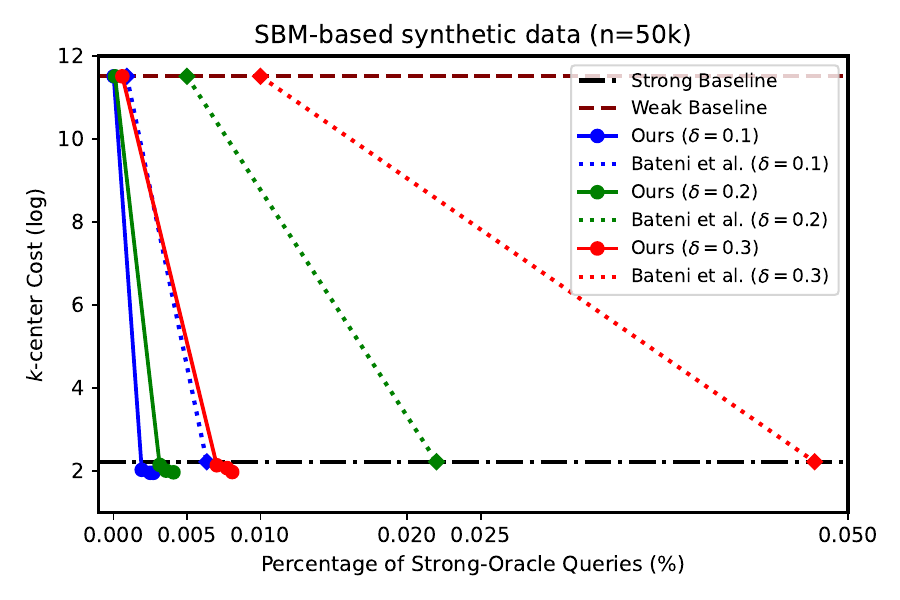}
        \label{fig:kc_sbm_50k}
    \end{minipage}

    \label{fig:kc_sbm_all}
\end{figure}

\noindent \textbf{Performance of $k$-center on MNIST dataset:} Next, we describe the performance of our $k$-center algorithm on MNIST dataset with SVD and $t$-SNE embeddings. Figures \ref{fig:kc_mnist_svd} and \ref{fig:kc_mnist_tsne} show how the clustering cost varies with the number of strong-oracle queries when these algorithms are run on MNIST datasets with SVD and $t$-SNE embeddings. We run these experiments for $\delta=0.1,0.2$ and $0.3$, and the results are given as blue, green and red lines in the Figures \ref{fig:kc_mnist_svd} and \ref{fig:kc_mnist_tsne}. The dash-dotted line corresponds to the baseline computed using the algorithm of \cite{G1986}. \cite{BDJW2024} does not provide any experimental results for the $k$-center problem on the MNIST dataset. In Figure~\ref{fig:kc_mnist_tsne}, we observe that the $k$-center cost with $\delta = 0.3$ is lower than the costs for $\delta = 0.1$ and $\delta = 0.2$. We believe this may be due to the nature of the t-SNE embeddings, where the distances between clusters are relatively small. As a result, the distance between two points in the same cluster can be larger than the distance between two points in different clusters. Because of this, the weak oracle might return the distance between two points in the same cluster with a smaller distance.

\begin{figure}[ht]
    \centering

    \begin{minipage}{0.5\linewidth}
        \captionsetup{font=small, justification=centering, width=0.9\linewidth}
        \caption{$k$-center on MNIST  with SVD embedding, $n = 60k$}
        \centering
        \includegraphics[width=\linewidth]{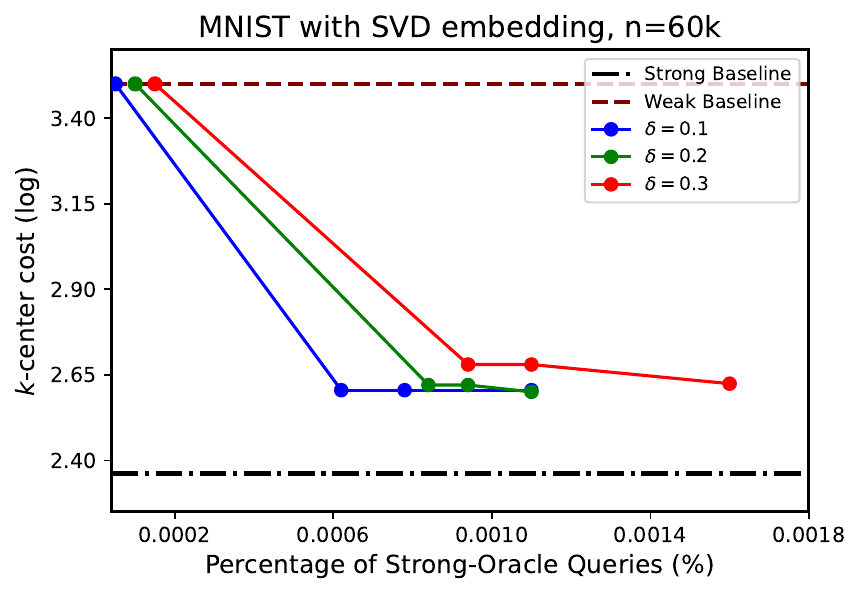}
        \label{fig:kc_mnist_svd}
    \end{minipage}%
    \hfill
    \begin{minipage}{0.5\linewidth}
        \captionsetup{font=small, justification=centering, width=0.9\linewidth}
        \caption{$k$-center on MNIST  with t-SNE embedding, $n = 60k$}
        \centering
        \includegraphics[width=\linewidth]{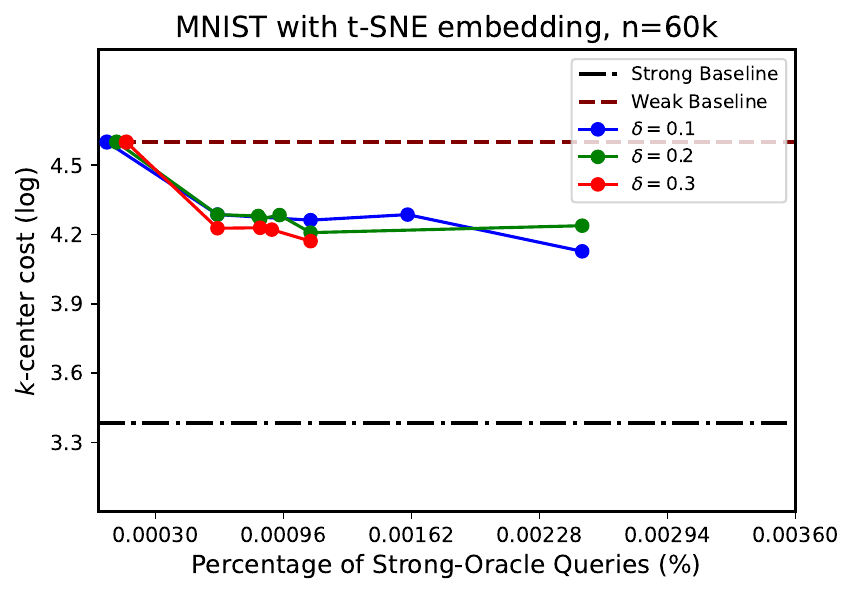}
        \label{fig:kc_mnist_tsne}
    \end{minipage}

    \label{fig:kc_svd-tsne}
\end{figure}

\begin{table}[H]
   \scriptsize
    \centering
    \caption{Variance of experimental results for $k$-center and $k$-means clustering problems for different datasets and failure probability values. The values mentioned in the table are to be multiplied by $10^{-3}$ to obtain the actual variance results.}
    \begin{tabular}{|c|p{1.2cm}|p{1.2cm}|p{1.2cm}|p{1.2cm}|p{1.2cm}|p{1.2cm}|}
        \hline
        Datasets & \multicolumn{3}{c|}{$k$-center} & \multicolumn{3}{c|}{$k$-means} \\ \hline
        & $\delta=0.1$ & $\delta=0.2$ & $\delta=0.3$ & $\delta=0.1$ & $\delta=0.2$ & $\delta=0.3$ \\ \hline
        
        SBM, n=10k     & $5.1$ & $3.2$ & $5.1$ & $51.4$ & $5.9$ & $0.77$ \\ \hline
        SBM, n=20k     & $0$ & $3.9$ & $1.0$ & $7.1$ & $6.3$ & $11$\\ \hline
        SBM, n=50k     & $7.8$ & $3.9$ & $7.8$ & $12.7$ & $11$ & $7.5$ \\ \hline
        MNIST, n=60k (t-SNE)  & $2.3$ & $0$ & $4.0$ & $4.4$ & $6.5$ & $20.0$\\ \hline
        MNIST, n=60k (SVD)  & $3.2$ & $35.0$ & $8.4$ & $0.3$ & $1.3$ & $1.0$ \\ \hline
    \end{tabular}
    
    \label{tab:variance}

\end{table}

Table \ref{tab:variance} shows the variance results for the experiments we conducted. These results are computed as follows. For both the $k$-means and $k$-center problems, we first fix the dataset and the value of $\delta$. Then, we vary  the constant  in the expressions used for the number of strong-oracle queries: $O\left(\frac{k^2\log^2 n}{(1/2 - \delta)^2}\right)$ for $k$-means, and $O\left(\frac{k^3\log^2 n}{(1/2 - \delta)^2}\right)$ for $k$-center.
We try different constant values between $0.0001$ and $10$, and choose the  value which minimizes the expression  $(\text{number of strong-oracle queries}) \times \log(\text{cost of clustering})$. For each of these choices, we run our algorithm five times. The average of the clustering costs from these runs is reported as the clustering cost, and the variance across the five runs is shown in Table \ref{tab:variance}.

\end{document}